\documentclass[10pt,twocolumn]{article}
\usepackage[letterpaper,textwidth=6.75in,textheight=9in,columnsep=0.25in,centering]{geometry}
\usepackage{titlesec}
\titleformat{\section}{\large\bfseries\raggedright}{\thesection}{0.6em}{}
\titlespacing*{\section}{0pt}{2.0ex plus 0.5ex minus 0.2ex}{1.5ex plus 0.3ex minus 0.2ex}
\titleformat{\subsection}{\normalsize\bfseries\raggedright}{\thesubsection}{0.6em}{}
\titlespacing*{\subsection}{0pt}{1.8ex plus 0.5ex minus 0.2ex}{0.8ex plus 0.2ex}
\titleformat{\subsubsection}{\normalsize\bfseries\raggedright}{\thesubsubsection}{0.6em}{}
\titlespacing*{\subsubsection}{0pt}{1.5ex plus 0.5ex minus 0.2ex}{0.5ex plus 0.2ex}
\usepackage[round]{natbib}
\usepackage[utf8]{inputenc}
\usepackage[T1]{fontenc}
\usepackage{amsmath,amssymb,amsthm}
\renewcommand{\eqref}[1]{Eq.~\textup{(\ref{#1})}}
\usepackage{mathtools}
\usepackage{bm}
\usepackage{graphicx}
\usepackage{booktabs}
\usepackage{tabularx}
\usepackage{multirow}
\usepackage{xcolor}
\usepackage{placeins}
\usepackage{float}
\usepackage{microtype}
\PassOptionsToPackage{hyphens}{url}
\usepackage[colorlinks=true,linkcolor=blue!50!black,citecolor=blue!50!black,urlcolor=blue!50!black]{hyperref}
\hypersetup{pdftitle={Certified Corruption Budgets: Anytime-Valid Leaderboard Claims under Adaptive Rigging},pdfauthor={Hamed Khosravi, Xiaoming Huo}}
\usepackage{aliascnt}
\usepackage[capitalize,noabbrev]{cleveref}
\usepackage{etoc}
\usepackage{algorithm}
\usepackage{algorithmic}

\newcommand{\E}{\mathbb{E}}
\renewcommand{\P}{\mathbb{P}}

\newcommand{\1}{\mathbf{1}}
\newcommand{\cF}{\mathcal{F}}
\newcommand{\cG}{\mathcal{G}}
\newcommand{\cH}{\mathcal{H}}
\newcommand{\Bhat}{\widehat{B}}
\newcommand{\Mp}{M_{+}}
\newcommand{\mm}{m_{-}}
\newcommand{\lb}{\bar\lambda}
\newcommand{\cins}{c_{\mathrm{ins}}}
\newcommand{\crep}{c_{\mathrm{rep}}}
\newcommand{\Bern}{\mathrm{Bern}}
\newcommand{\Bin}{\mathrm{Bin}}
\newcommand{\KL}{\mathrm{KL}}
\newcommand{\classP}{\mathsf{P}}
\newcommand{\classV}{\mathsf{V}}
\newcommand{\classO}{\mathsf{O}}

\theoremstyle{plain}
\newtheorem{theorem}{Theorem}[section]
\newaliascnt{proposition}{theorem}
\newtheorem{proposition}[proposition]{Proposition}
\aliascntresetthe{proposition}
\newaliascnt{lemma}{theorem}
\newtheorem{lemma}[lemma]{Lemma}
\aliascntresetthe{lemma}
\newaliascnt{corollary}{theorem}
\newtheorem{corollary}[corollary]{Corollary}
\aliascntresetthe{corollary}
\crefname{theorem}{Theorem}{Theorems}
\crefname{proposition}{Proposition}{Propositions}
\crefname{lemma}{Lemma}{Lemmas}
\crefname{corollary}{Corollary}{Corollaries}
\crefname{definition}{Definition}{Definitions}
\theoremstyle{definition}
\newtheorem{definition}{Definition}[section]

\theoremstyle{remark}

\title{\textbf{Certified Corruption Budgets:\\ Anytime-Valid Leaderboard Claims under Adaptive Rigging}}
\author{
\begin{tabular}{c@{\hspace{3em}}c}
Hamed Khosravi & Xiaoming Huo\\
\texttt{hkhosravi7@gatech.edu} & \texttt{huo@gatech.edu}
\end{tabular}\\[8pt]
H.~Milton Stewart School of Industrial and Systems Engineering\\
Georgia Institute of Technology, Atlanta, GA, USA
}
\date{}

\begin{document}

\maketitle

\begin{abstract}
Public leaderboards for AI models are read continuously, and attackers can see every published standing. Vote rigging, selective disclosure of private variants, and benchmark contamination can each move a ranking. Existing guarantees assume genuine records or bound the corruption per step, which an attacker who corrupts in bursts evades. We introduce the \emph{certified corruption budget}, a tolerance $\Bhat_t$ computed after $t$ records and published with each pairwise claim. With probability at least $1-\alpha$, simultaneously at all times, the claim is correct or more than $\Bhat_t$ records were corrupted. It holds against attackers who watch every certificate, with no bound on their budget. Forged records and records altered once seen require different certificates: the certificate for forgeries fails, with probability approaching one, against an attacker who flips votes it has seen, while one that charges roughly twice as much per record remains valid, with constant bets even against attackers who see the future, and no smaller charge is valid at every level. The certified budget grows nearly as fast as any valid method allows: with a win fraction $\frac12+\delta$, each new record adds close to $2\delta$ to the number of forged records the claim can withstand ($\delta$ flipped). Publishing the best of $V$ private variants costs only an amount growing like $\log V$. In replays on $1.8$ million Chatbot Arena votes, a few hundred rigged votes make standard confidence intervals certify false orderings, while ours stays valid. On real votes, our certificate shows that clearly separated models withstand about $2{,}000$ forged votes.
\end{abstract}

\etocdepthtag.toc{mtmain}
\section{Introduction}\label{sec:intro}

Public leaderboards decide which artificial intelligence (AI) models are trusted, bought and built on. They are read while being computed. A \emph{leaderboard} ranks models from performance records. On Chatbot Arena, each record is a \emph{battle}: a user prompts two anonymous models and votes for the better answer or a tie. Its public file has nearly two million votes \citep{zheng2023judging,chiang2024arena,chatbotarena_data}. On static benchmarks such as HELM, the Holistic Evaluation of Language Models, a record compares two models on one test item \citep{liang2023helm}.

Leaderboards are attacked, and the attackers see what is published. In simulations of Chatbot Arena, $381$ rigged votes (votes cast to favor a chosen model) moved a low-ranked model up four places, and anonymous models can be identified to aim such votes \citep{huang2025voting}. Some developers privately test many variants of a model and publish only the best. One developer tested 27 \citep{singh2025illusion}. On static benchmarks, training on the test items (\emph{contamination}) inflates a model's score \citep{sainz2023contamination,xu2024contamination}. In each case the attacker can wait, see which comparisons are close and strike there.

The usual statistical guarantees do not address such an attacker. Confidence intervals for ratings and ranks assume genuine records analyzed once, at a sample size fixed in advance \citep{chiang2024arena,mogstad2024}. Methods valid under analysis after every record also assume genuine records \citep{ramdas2023gtsa}, and their robust versions allow a fixed fraction of corruption at every step \citep{saha2024,wang2023huber}, which a \emph{burst} (many corrupted records close together) violates. A leaderboard's reader needs a different statement: not ``model $i$ leads model $j$'' but ``model $i$ leads model $j$ unless more than a stated number of records were corrupted''.

We attach to each leaderboard claim a certified corruption budget, which we call its \emph{corruption tolerance}: a number of corrupted records that the claim is guaranteed to withstand. With each claim about two models $i$ and $j$, the platform publishes this number, $\Bhat_t$, computed from the first $t$ records on the pair. Except with a small probability $\alpha$ chosen by the platform (say $5\%$), and no matter when the claim is read, either the claim holds or more than $\Bhat_t$ of those records were forged or altered. We certify two claims. The \emph{lead claim}, that one model leads another on average so far, is the usual interpretation of a leaderboard standing. The \emph{edge claim} is that one model has had an advantage over another at some point. The two coincide when win chances are stable. The guarantee holds even if the attacker watches every published tolerance, and assumes no limit on the number of corrupted records.

We compute the tolerance by betting on the claim. It is the evidence won beyond what a test on genuine records needs, divided by a \emph{cost} per record. This cost is the most evidence one corrupted record can fake. The construction has four properties. \emph{First}, the tolerance counts corrupted records, the unit in which attacks are measured, so it compares directly with known attacks. \emph{Second}, it separates forging a vote from changing a seen one, as a colluding annotator or a contaminated test item does. A certificate for forgeries fails against altered votes. Roughly doubling the cost per record restores the guarantee. With constant bets it holds even against an attacker knowing the future. \emph{Third}, for independent and identically distributed (i.i.d.)\ records without ties, no valid method certifies substantially more corrupted records from the same data. \emph{Fourth}, it prices publishing the best of many private variants and updates cheaply per vote.

\textbf{Contributions.} (i) We introduce corruption tolerance, the number of corrupted records a claim withstands, and three attacker classes. (ii) We show that no existing method covers this setting and give a procedure that needs no bound on the number of corrupted records. (iii) We prove these guarantees for both claims and extend them to whole leaderboards. (iv) We confirm the main results by simulation and on real data. On $1.8$ million Chatbot Arena votes, well-separated models are certified against about $2{,}000$ forged votes, but GPT-4o (2024-05-13) is not certified above Claude-3.5-Sonnet when all pairs are claimed at once. On five frontier benchmarks, no top-two ordering is certified when all pairs are claimed at once.

\section{Problem Formulation}\label{sec:formulation}

\subsection{Records, claims and the null hypothesis}

The data for a claim about models $i$ and $j$ are a stream of records. We extend the formulation to whole leaderboards in \cref{sec:simultaneous}. A \emph{record} is one comparison of $i$ with $j$ together with its outcome: $1$ if $i$ wins, $0$ if $j$ wins and $\tfrac12$ for a tie. Records are indexed by $t=1,2,\dots$ in the order the platform receives them, and it may choose which pair to compare next using everything it has seen. After each record it publishes a certificate for the claim (\cref{sec:tolerance}).

We model the stream as random variables on one probability space. The randomness comes from the platform's normal process (users, prompts, model answers and votes) and from an \emph{attacker}, who may randomize. The attacker corrupts records in two ways. It may \emph{insert} a forged record that the normal process never produced, such as a vote from a fake account. It may also \emph{replace} a record by overwriting the value that the normal process produced, as a colluding annotator or a contaminated test item does. A replacement of a loss by a win is a \emph{flip}. A record is \emph{corrupted} if it is inserted or replaced, and \emph{uncorrupted} otherwise.

Four random variables describe the $t$th record, and the platform observes only the first. The \emph{observed outcome} $W_t\in[0,1]$ is the outcome that the platform receives. The indicator $I_t$ equals $1$ if the record is inserted, and $C_t$ equals $1$ if it is corrupted. The \emph{uncorrupted value} $X_t\in[0,1]$ is the value that the normal process produced. It exists when $I_t=0$, and we set $X_t=\frac12$ when $I_t=1$ (formal model in \cref{app:model}). Thus $W_t=X_t$ for an uncorrupted record, while the attacker chooses $W_t$ for a corrupted record, and a replaced record hides its uncorrupted value. The number of corrupted records among the first $t$ is the random variable $N_t=\sum_{s=1}^{t}C_s$. The normal process and the attacker's strategy together determine the joint distribution of these variables. We write $\P$ for probability and $\E$ for expectation under this distribution. The time index $t$ and the error level $\alpha$ are fixed numbers, not random variables.

The \emph{edge claim}, that $i$ has ever had an edge over $j$, is false when, at every record with an uncorrupted value, $i$ is expected to do no better than $j$ given the past. Let the $\sigma$-field $\cG_{t-1}$ collect everything that happened before the $t$th record's value is produced. It includes the variables of all earlier records on any pair, all published certificates, the attacker's private randomization and $I_t$. The null hypothesis is
\begin{equation}\label{eq:null}
\begin{aligned}
&H_0^{ij}:\quad \text{for every } t\ge1,\\
&\E[X_t\mid\cG_{t-1}]\le\tfrac12 \ \text{ on the event } \{I_t=0\}.
\end{aligned}
\end{equation}
The conditional expectation is random, and the inequality must hold with probability one at each non-inserted record. The hypothesis $H_0^{ij}$ is a property of the joint distribution, not a random event. It is one statement about the whole stream and has no index $t$. With ties coded as $\tfrac12$, it says that at every such record, given the past, $i$ is no more likely to win than to lose.

The null hypothesis assumes neither independence nor a constant win probability, so the models, users and prompts may change over time. The edge claim is \emph{correct} when $H_0^{ij}$ fails, that is, when the event $\{I_t=0,\ \E[X_t\mid\cG_{t-1}]>\tfrac12\}$ has positive probability for at least one $t$. A correct claim thus says that $i$ was favored over $j$ at one record or more (discussed in \cref{app:testing}). If $\E[X_t\mid\cG_{t-1}]$ equals the same constant $\mu$ at every non-inserted record (a stable mean outcome), the claim is correct exactly when $\mu>\tfrac12$. This holds for i.i.d.\ uncorrupted values and under the \emph{Bradley--Terry} model \citep{bradley1952,chiang2024arena} that Arena fits (win probabilities are logistic in score differences). This stable case is the usual reading of a leaderboard claim.

The \emph{lead claim} concerns the records so far. Write $\mu_s=\E[X_s\mid\cG_{s-1}]$ for the conditional mean of an uncorrupted value and let
\begin{equation}\label{eq:delta}
\Delta_t=\sum_{s\le t}(1-I_s)\big(\mu_s-\tfrac12\big)
\end{equation}
be the \emph{uncorrupted margin} accumulated over the first $t$ records. The lead claim ``$i$ leads $j$'' is correct at time $t$ when $\Delta_t>0$, that is, when the normal process favored $i$ on average over these records. It needs no stable mean, is what a standing usually asserts, and coincides with the edge claim in the stable case. We certify both.

\subsection{Three classes of attackers}\label{sec:attackers}

The attacker classes differ in what the attacker knows when it acts. Every attacker is \emph{adaptive}: it observes every published certificate and may choose its next action from what it has seen. No limit is placed on how many records it corrupts.
\begin{itemize}
\item[$\classP$] \emph{Predictable attackers} decide whether to corrupt the $t$th record, and what its value is, before its uncorrupted value is produced. Every insertion is of this kind, and so is a replacement that discards the uncorrupted value unseen.
\item[$\classV$] \emph{Value-dependent attackers} may in addition see the uncorrupted value $X_t$ and then decide whether to replace it, before the next record arrives.
\item[$\classO$] \emph{Offline attackers} see the entire sequence of uncorrupted values, including future ones, before acting. Since their actions then depend on the future, the null for $\classO$ is stated for the uncorrupted values alone. These form a sequence, indexed by $u=1,2,\dots$ in order of arrival, whose law does not depend on the attacker and satisfies $\E[X_u\mid X_1,\dots,X_{u-1}]\le\frac12$.
\end{itemize}
The classes are nested in what the attacker knows, and \cref{sec:theory} gives the price of guarding against each.

\subsection{Corruption tolerance}\label{sec:tolerance}

\begin{definition}[Corruption tolerance]\label{def:budget}
Let $(\Bhat_t)_{t\ge1}$ take values in $\{-1,0,1,2,\dots\}$, with each $\Bhat_t$ a function of the first $t$ observed outcomes $W_1,\dots,W_t$. For $\alpha\in(0,1)$, the sequence is an \emph{$\alpha$-valid corruption tolerance for the edge claim} against an attacker class if, for every normal process and attacker in the class under which $H_0^{ij}$ holds,
\begin{equation}\label{eq:validity}
\P\big(\exists\, t\ge1:\ \Bhat_t\ge N_t\big)\le\alpha .
\end{equation}
It is \emph{$\alpha$-valid for the lead claim} against the class if, for every normal process and attacker in it,
\begin{equation}\label{eq:validity-avg}
\P\big(\exists\, t\ge1:\ \Delta_t\le0\ \text{and}\ \Bhat_t\ge N_t\big)\le\alpha .
\end{equation}
\end{definition}

For $\classO$ the condition that $H_0^{ij}$ holds is replaced by the condition that the uncorrupted values satisfy $\E[X_u\mid X_1,\dots,X_{u-1}]\le\frac12$, the null for offline attackers of \cref{sec:attackers}. Under $H_0^{ij}$ the uncorrupted margin $\Delta_t$ is at most $0$ at every $t$. Hence \eqref{eq:validity-avg} implies \eqref{eq:validity}, and \eqref{eq:validity-avg} makes no hypothesis on the law. \Cref{def:budget} controls the type I error of a family of sequential tests. At time $t$ the platform publishes the \emph{certificate} ``$i$ has ever had an edge over $j$, unless more than $\Bhat_t$ of the first $t$ records were corrupted''. The number $\Bhat_t$ is the claim's corruption tolerance at that time. Nothing is claimed when $\Bhat_t=-1$. For each integer $b\ge0$ the event $\Bhat_t\ge b$ rejects the hypothesis ``$H_0^{ij}$ holds and $N_t\le b$''. When $H_0^{ij}$ holds, a true hypothesis of this family is rejected exactly when $\Bhat_t\ge N_t$. So \eqref{eq:validity} bounds by $\alpha$ the probability of a \emph{false statement}. A false statement is a type I error at any time and for any $b$. When $H_0^{ij}$ fails, no type I error is possible (details in \cref{app:testing}). The platform does not observe $N_t$. It therefore publishes the conditional statement and not a rejection of $H_0^{ij}$.

A tolerance, or the method that computes it, is \emph{valid} against an attacker class if it is $\alpha$-valid for the edge claim against that class. Validity for the lead claim is always named. The aim is the largest valid tolerance. The claim is \emph{certified against} $b$ corrupted records at time $t$ if $\Bhat_t\ge b$, and \emph{not certified} if $\Bhat_t=-1$. The attacker's \emph{budget} is the number of records it may corrupt. The claim's tolerance is $\Bhat_t$. Raising $\Bhat_t$ only makes \eqref{eq:validity} harder to satisfy, and the trivial tolerance $-1$ is valid and claims nothing. For i.i.d.\ records without ties, our tolerance grows at nearly the fastest rate that any valid tolerance can reach. \Cref{sec:rates} proves this.

The guarantee is \emph{anytime-valid}. It holds at all times simultaneously, so anyone may act on the statement at a time of their choosing, even one chosen by looking at the data. With $b=0$ it is an ordinary anytime-valid test of $H_0^{ij}$ on uncorrupted records.

\section{Related Work}\label{sec:related}

Anytime-valid methods test ranking claims under genuine data, and robust methods under a prespecified corruption level. Neither infers from the data how much adaptive corruption a claim withstands, which corruption tolerance supplies.

\textbf{Anytime-valid inference assumes genuine data.} A nonnegative process whose expectation stays at most one under the null, at every stopping time, is an \emph{e-process}. By Ville's inequality it ever exceeds $1/\alpha$ with probability at most $\alpha$ \citep{ville1939,ramdas2023gtsa}. \emph{Testing by betting} builds e-processes as the wealth of a gambler who starts with one unit and bets against the null at fair odds \citep{shafer2021,vovk2021evalues,grunwald2024safe,ramdas2023gtsa}. Confidence sequences \citep[intervals valid uniformly over time;][]{darling1967,howard2021} can be built this way for bounded means \citep{waudbysmith2024}, and so can two-sample tests \citep{shekhar2023}. They suit leaderboards that are read continuously \citep{khosravi2026rank,gao2026bbedge,gu2026online}, but one forged record can multiply the wealth by a constant, and a burst can push it past any threshold.

\textbf{Robust sequential tests bound corruption per step.} \citet{huber1965} made the sequential probability ratio test \citep{wald1945} robust by clipping the likelihood ratio for data with independent outliers. \citet{saha2024} give anytime-valid tests in which, at every step, an adaptive adversary picks any conditional law within total-variation distance $\varepsilon$ of the null, with $\varepsilon$ fixed in advance. \citet{wang2023huber} build confidence sequences under the same bound, stated for i.i.d.\ data for simplicity. They regard replacement by an adversary that has seen all observations as inherently fixed-time. Without a known contamination bound, nontrivial anytime-valid confidence intervals do not exist \citep{agrawal2024crimed}. Related work covers online change points \citep{tang2026} and, at a fixed sample size, unknown outlier fractions \citep{luo2024adaptive}. No such sequential test counts corrupted records, and none can describe a burst.

\begin{proposition}[Per-step bounds cannot describe bursts]\label{prop:perstep}
Suppose at every step the conditional law of a binary outcome may be any law within total-variation distance $\varepsilon$ of some Bernoulli law $\Bern(p)$ (outcome $1$ with probability $p$) with $p\le\frac12$. A forged win (the point mass at $1$) is allowed only if $\varepsilon\ge\frac12$. Then every Bernoulli law is allowed under the null, so every valid test has power at most $\alpha$ against every alternative.
\end{proposition}

Leaderboard inference assumes genuine votes, and attack studies and certified defenses give no sequential certificate (\cref{app:fields}).

\section{Method: Deflated Capped Betting}\label{sec:method}

\subsection{Capped betting wealth}

The tolerance is built from the wealth of a gambler who bets on $i$ at every record, with a cap on the bet. Fix a \emph{cap} $\lb\in(0,2)$ and, for each record, a \emph{bet} $\lambda_t\in[0,\lb]$ chosen from the first $t-1$ observed records. A \emph{constant bet} is the same at every record, $\lambda_t=\lb$ for all $t$, and we then write $\lambda$ for $\lb$. The wealth after $t$ records is
\begin{equation}\label{eq:wealth}
K_t=\prod_{s=1}^{t}\big(1+\lambda_s(W_s-\tfrac12)\big),\qquad K_0=1 .
\end{equation}
Each factor lies between $\mm=1-\lb/2$, the smallest value one record can contribute, and $\Mp=1+\lb/2$, the largest. A tie gives the factor $1$. If all records are uncorrupted and $H_0^{ij}$ holds, $K_t$ is a nonnegative supermartingale, so ${\P(\exists\,t:\ K_t\ge1/\alpha)}\le\alpha$ by Ville's inequality. If $i$ wins more often than it loses and the bet is small enough, $K_t$ grows exponentially.

The cap is what makes robustness possible. Whatever an attacker does to one record, it changes that record's factor by a bounded ratio. A forged record contributes a factor of at most $\Mp$. A replaced record changes its factor from at least $\mm$ to at most $\Mp$, a ratio of at most $\Mp/\mm$. We call
\begin{equation}\label{eq:costs}
\begin{aligned}
\cins(\lb)&=\log\Mp=\log\!\big(1+\tfrac{\lb}{2}\big),\\
\crep(\lb)&=\log\frac{\Mp}{\mm}=\log\frac{2+\lb}{2-\lb}
\end{aligned}
\end{equation}
the \emph{insertion cost} and the \emph{replacement cost} of a record, measured in \emph{nats} (units of natural logarithm of wealth). For small caps, $\crep\approx 2\cins\approx\lb$.

\textbf{The corruption tolerance} is the number of corrupted records that the evidence could have absorbed. With cost $c\in\{\cins(\lb),\crep(\lb)\}$ we report
\begin{equation}\label{eq:bhat-single}
\Bhat_t=\max\Bigl\{-1,\ \Bigl\lfloor \frac{\log(\alpha K_t)}{c}\Bigr\rfloor\Bigr\}.
\end{equation}
In words, the evidence $\log K_t$ exceeds the threshold $\log(1/\alpha)$ of a test on uncorrupted data by $\log(\alpha K_t)$ nats, and each corrupted record supplies at most $c$ of them. The insertion cost is valid against predictable attackers ($\classP$), and the replacement cost against all classes ($\classO$ with constant bets; \cref{thm:validity}).

\subsection{A grid of bets}

No single bet works well everywhere, so we run several bets and charge each its own cost before combining them. The best bet depends on the unknown \emph{margin}, the amount by which $i$'s mean uncorrupted outcome exceeds $\tfrac12$. Smaller bets reach larger tolerances in the long run but need more records to start (\cref{sec:rates}). We therefore run $m$ gamblers in parallel, indexed by $k=1,\dots,m$. Gambler $k$ has its own cap $\lb_k\in(0,2)$, decreasing from $\lb_1$ to $\lb_m$, and makes the constant bet $\lb_k$ at every record (\emph{bet $k$}). These caps form a \emph{grid}, a finite set of bet sizes spread so that one of them is close to the best bet for a wide range of margins. The gamblers have weights $w_k>0$ that satisfy $\sum_k w_k\le1$. In all experiments $\lb_k=2^{-(k-1)/2}$ for $k=1,\dots,14$, from $1$ down to $0.011$, and $w_k=1/14$. Let $K^{(k)}_t$ be the wealth of \eqref{eq:wealth} of gambler $k$, and write $c_k=c(\lb_k)$ for the chosen cost of \eqref{eq:costs} at the cap $\lb_k$. A wealth is \emph{deflated} by $b$ records when it is divided by $e^{b\,c_k}$, the most that $b$ corrupted records could have multiplied it by ($e$ is the base of the natural logarithm). The \emph{deflated average} is
\begin{equation}\label{eq:deflated}
\Phi_t(b)=\sum_{k=1}^m w_k\,K^{(k)}_t\,e^{-b\,c_k},\qquad b\ge0.
\end{equation}
We report
\begin{equation}\label{eq:bhat-grid}
\Bhat_t=\max\big\{b\in\{0,1,2,\dots\}:\ \Phi_t(b)\ge1/\alpha\big\},
\end{equation}
and $\Bhat_t=-1$ if $\Phi_t(0)<1/\alpha$. Here $b$ is a candidate value of $N_t$. With one bet ($m=1$, $w_1=1$) this is \eqref{eq:bhat-single}. Our tolerance is never smaller than that of the best single bet run at level $\alpha w_k$. As $\Phi_t$ decreases in $b$, bisection finds $\Bhat_t$. The procedure (\cref{alg:ccb}, \cref{app:algorithm}) stores $m$ numbers per pair and needs $O(m\log t)$ operations per record.

\subsection{An exponential wealth for the lead claim}\label{sec:expwealth}

The lead claim needs a wealth with a logarithm linear in the outcome: for a bet size $\gamma\in(0,4)$ let
\begin{equation}\label{eq:zwealth}
Z^{(\gamma)}_t=\exp\Big(\gamma\sum_{s\le t}\big(W_s-\tfrac12\big)-\frac{\gamma^2t}{8}\Big),
\end{equation}
with the costs $\cins^{Z}(\gamma)=\gamma/2-\gamma^2/8$ and $\crep^{Z}(\gamma)=\gamma$, both positive. A forged win raises $\log Z^{(\gamma)}_t$ by exactly $\gamma/2-\gamma^2/8$ and a replaced record by at most $\gamma$, hence the costs, and $\gamma^2t/8$ is the price of betting $\gamma$ on outcomes in $[0,1]$ \citep{hoeffding1963}. The \emph{lead tolerance} $\Bhat^Z_t$ is \eqref{eq:bhat-grid} with $Z^{(\gamma_k)}_t$ for $K^{(k)}_t$, $\gamma_k$ for the caps $\lb_k$, and $c^Z\in\{\cins^Z,\crep^Z\}$ for the cost $c$.

\section{Theoretical Guarantees}\label{sec:theory}

\subsection{Validity against adaptive attackers}

\begin{theorem}[Validity]\label{thm:validity}
Let $\Bhat_t$ be the tolerance of \eqref{eq:bhat-grid} with constant bets $\lb_1,\dots,\lb_m\in(0,2)$ and weights $w_k>0$, $\sum_k w_k\le1$.
(a) With the insertion cost $c=\cins$, $\Bhat_t$ is valid at level $\alpha$ against $\classP$. (b) With the replacement cost $c=\crep$, it is valid against $\classV$. (c) With $c=\crep$, it is valid against $\classO$ under the null of $\classO$ (\cref{sec:attackers}).
Parts (a) and (b) remain true if gambler $k$ replaces its constant bet $\lb_k$ by any bet $\lambda_{k,t}\in[0,\lb_k]$ computed from the first $t-1$ observed records.
\end{theorem}

In \cref{thm:validity} the attacker may corrupt any number of records, watch every published tolerance and spend its whole effort in one burst. The cheaper insertion cost suffices against forged records, and the replacement cost is required when records may be altered after being seen. Part (c) covers contamination by a developer who knows every test item in advance.

\subsection{Neither cost can be lowered}

\begin{theorem}[Sharpness]\label{thm:sharp}
Consider a single constant bet $\lambda\in(0,2)$, write $\cins=\cins(\lambda)$ and $\crep=\crep(\lambda)$, let $B$ be the attacker's budget, and take uncorrupted values that are independent $\Bern(\frac12)$ and the tolerance of \eqref{eq:bhat-single} with an arbitrary cost $c>0$.
(a) If $c<\cins$, an attacker in $\classP$ that inserts $B\ge\log(1/\alpha)/(\cins-c)$ forged wins before the first uncorrupted record causes a false statement with probability one.
(b) If $c<2\cins$, the attacker in $\classV$ that replaces each of the first $B$ uncorrupted losses by a win causes a false statement with probability tending to one as $B\to\infty$. In particular the insertion cost is invalid against $\classV$.
(c) If $2\cins\le c<\crep$, the attacker in $\classV$ that replaces every uncorrupted loss by a win causes a false statement with probability at least $\alpha^{\theta}e^{-\theta\cins}$, where $\theta\in[0,1)$ is the largest root of $e^{\theta\cins}+e^{\theta(\cins-c)}=2$ (positive if and only if $c>2\cins$). This exceeds $\alpha$ for all small enough $\alpha$, so no cost below $\crep$ is valid at every level.
(d) Among attackers in $\classO$ that make no insertions and at most $B$ replacements, the attacker in (b) maximizes $K_t$ simultaneously for every $t$, and maximizes the probability of a false statement for every $c\le\crep$. The same holds for the grid certificate of \eqref{eq:bhat-grid} with any costs $0<c_k\le\crep(\lb_k)$.
\end{theorem}

By \cref{thm:sharp}, what makes an attacker dangerous is acting on the values it sees, not seeing the future. Charging only the insertion cost, valid for predictable corruption, assumes the attacker cannot pick which uncorrupted records to corrupt. Flipping seen votes breaks this assumption. Both costs are \emph{sharp} for a constant bet: none smaller is valid at every level.

\subsection{How fast the tolerance grows}\label{sec:rates}

For i.i.d.\ uncorrupted values without ties, the tolerance grows nearly as fast as any valid tolerance. Suppose no attack and i.i.d.\ uncorrupted values, so that $W_t=X_t$, with $\P(W=1)=a$, $\P(W=0)=\ell<a$ and ties otherwise. The margin, the amount by which $i$'s mean outcome exceeds $\frac12$ (\cref{sec:method}), is $\delta=(a-\ell)/2>0$, and $\pi=a+\ell$ is the probability of no tie. The expected log-wealth growth per record under the constant bet $\lambda$ is
\begin{equation}\label{eq:growth}
g(\lambda)=\E\log\!\big(1+\lambda(W-\tfrac12)\big)=\delta\,\crep(\lambda)-\tfrac{\pi}{2}\kappa(\lambda),
\end{equation}
where $\kappa(\lambda)=-\log(1-\lambda^2/4)\ge0$. Let $\tau_b={\inf\{t:\Bhat_t\ge b\}}$ be the number of records needed to reach tolerance $b$, and write $x^+=\max(x,0)$.

\begin{theorem}[Growth rate and time to certify]\label{thm:upper}
For such i.i.d.\ uncorrupted values and no attack, the tolerance of \eqref{eq:bhat-grid} with either cost $c\in\{\cins,\crep\}$ satisfies
\begin{equation}\label{eq:rate}
\lim_{t\to\infty}\frac{\Bhat_t}{t}=\rho:=\max_{k}\frac{g(\lb_k)^+}{c(\lb_k)},
\end{equation}
almost surely and, for every integer $b\ge0$, with $\min\emptyset=+\infty$,
\begin{equation}\label{eq:tau-upper}
\E\,\tau_b\le\min_{k:\,g(\lb_k)>0}\frac{\log\frac{1}{\alpha w_k}+b\,c(\lb_k)+\cins(\lb_k)}{g(\lb_k)} .
\end{equation}
Moreover $g(\lambda)/\cins(\lambda)\ge2\delta-\pi\lambda/2$ and $g(\lambda)/\crep(\lambda)\ge\delta-\pi\lambda/6$ for $\lambda\le1$.
\end{theorem}

\begin{theorem}[Lower bound]\label{thm:lower}
Let $\delta\in(0,\frac12)$ and let $\Bhat_t$ be any $\alpha$-valid corruption tolerance against $\classP$ (respectively $\classV$). Let $q=2\delta$ (respectively $q=\delta$). If the uncorrupted values are independent $\Bern(\frac12+\delta)$ and there is no attack, then for all $n\ge1$ and real $b\ge0$,
\begin{align}
\P\big(\textstyle\max_{t\le n}\Bhat_t\ge b\big)&\le\alpha+\P\big(\Bin(n,q)>b\big),\label{eq:lower-coupling}\\
\P\big(\textstyle\max_{t\le n}\Bhat_t\ge 0\big)&\le\frac{n\,\KL_\delta+\log2}{\log(1/\alpha)},\label{eq:lower-kl}
\end{align}
where $\Bin(n,q)$ is a binomial variable and $\KL_\delta=\KL(\Bern(\frac12+\delta)\,\|\,\Bern(\frac12))\le4\delta^2$ is a Kullback--Leibler divergence. Consequently, reaching an integer tolerance $b$ by time $n$ with probability at least $1-\beta$, for some $\beta\in(0,1-\alpha)$, requires $n\ge(1-\alpha-\beta)(b+1)/q$ and $n\ge[(1-\beta)\log(1/\alpha)-\log2]/(4\delta^2)$. Moreover ${\P(\Bhat_n\ge(q+\eta)n)}\le\alpha+e^{-2n\eta^2}$ for every $\eta>0$.
\end{theorem}

\Cref{thm:upper,thm:lower} show that the growth rate approaches the optimum. By the last claim of \cref{thm:upper} and by \cref{thm:lower}, $\rho$ lies between $q-\lb_m/2$ (insertion, $q=2\delta$) or $q-\lb_m/6$ (replacement, $q=\delta$) and $q$ when $\lb_m\le1$, so $\rho\to q$ as $\lb_m\to0$. The time to reach tolerance $b$ is of order $\log(1/\alpha)/\delta^2+b/\delta$, the price of a sequential test plus a robustness price linear in $b$, with upper and lower bounds matching up to constants (\cref{cor:explicit,cor:lower-mean} in \cref{app:proofs}). Faking a margin $\delta$ takes forging a fraction $2\delta$ of the records or flipping a fraction $\delta$, and the tolerance reaches nearly that fraction.

\begin{figure*}[t]
\centering
\includegraphics[width=\textwidth]{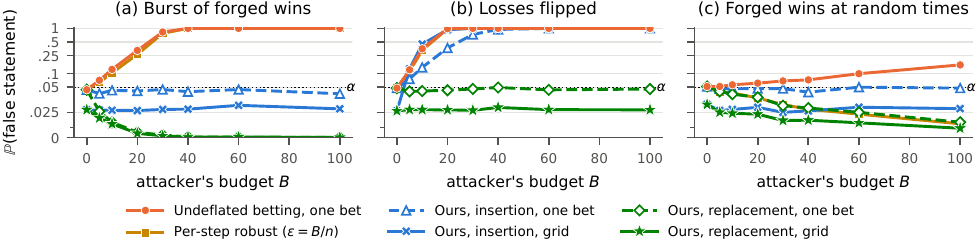}
\caption{False-statement probability in $n=5{,}000$ null records (win chance $\frac12$) versus the budget $B$ ($10{,}000$ runs; dotted: $\alpha$). (a) $B$ forged wins in one burst, timed to break undeflated betting. (b) The first $B$ uncorrupted losses flipped, the best offline replacement attack (\cref{thm:sharp}(d)). (c) $B$ forged wins, each inserted before a record with probability $B/n$. Markers hide the $95\%$ Monte Carlo intervals (half-width $\le0.01$). Baselines: \cref{sec:experiments}.}
\label{fig:validity}
\end{figure*}

\begin{figure*}[t]
\centering
\includegraphics[width=\textwidth]{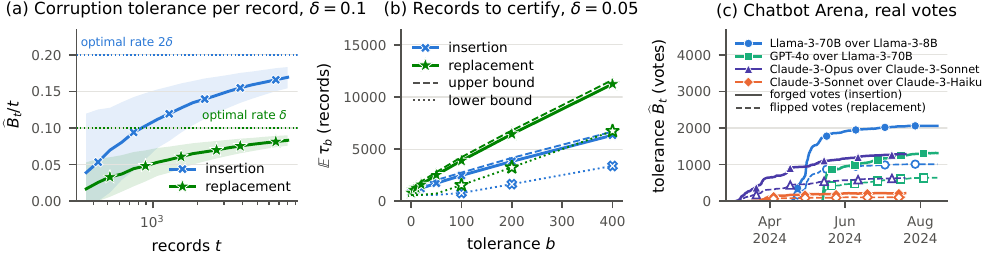}
\caption{(a) Median and $10$--$90\%$ band of $\Bhat_t/t$ at margin $\delta=0.1$ ($2{,}000$ repetitions), with the optimal rates $2\delta$ and $\delta$. (b) Mean records to reach tolerance $b$ at $\delta=0.05$, with the upper bound of \cref{thm:upper} and the lower bound of \cref{cor:lower-mean}. (c) Tolerance of four edge claims on real, time-ordered Chatbot Arena votes.}
\label{fig:rates}
\end{figure*}

\subsection{Selection and whole leaderboards}\label{sec:simultaneous}

\begin{proposition}[Price of selection]\label{prop:selection}
Suppose a developer evaluates $V$ variants of model $i$ against model $j$, and the variant that each record concerns is fixed before its value is produced. Call a variant \emph{null} if its records against $j$ satisfy the null hypothesis of \eqref{eq:null}, and suppose that the null hypothesis of each null variant holds given the records of all variants. Let the published variant $\hat v_t$ be chosen by any rule, possibly using all records. The grid certificate of \eqref{eq:bhat-grid} computed for $\hat v_t$ with $1/\alpha$ replaced by $V/\alpha$ makes a false statement about a null variant, at any time, with probability at most $\alpha$, against the attacker classes and costs of \cref{thm:validity}(a,b). Conversely, suppose that $\alpha\le\frac12$, the $V$ variants have independent $\Bern(\frac12)$ records and one constant bet $\lambda$ is used. Then replacing $V/\alpha$ by $V'/\alpha$ with any $0<V'\le V/(2\Mp)$, where $\Mp=1+\lambda/2$, gives a false statement with probability greater than $\alpha$ for some selection rule, even without an attack.
\end{proposition}

For the best of at most $V$ private variants, with $V$ fixed in advance and variants possibly created as results come in, $1/\alpha$ must therefore scale with $V$, and $V/\alpha$ suffices. Selection costs between $\log V-\log(2\Mp)$ and $\log V$ nats, about $(\log V)/c$ records: $35$ forged records for the $V=27$ variants of \citet{singh2025illusion} and $\lambda=0.2$, a quantitative argument for registering variants in advance.

A leaderboard makes many claims at once, and the attacker spends one budget across all of them. Let $N_t$ now count the corrupted records on all pairs up to (global) time $t$, and let $\Phi^{ij}_t(b)$ be the deflated average of \eqref{eq:deflated} for the ordered pair $(i,j)$. A \emph{joint claim} is a set $D$ of ordered pairs $(i,j)$ for which the edge claim is asserted. It is false if $H_0^{ij}$ holds for some $(i,j)\in D$.

\begin{proposition}[Simultaneous certificates]\label{prop:closed}
Assume that the pair of each record is fixed before its value is produced and that each true null $H_0^{ij}$ holds given the records on all pairs. Let $\mathcal{S}$ be a family of sets of ordered pairs that contains the set $S^\ast$ of pairs whose null hypotheses are true. For a joint claim $D$ that intersects some $S\in\mathcal{S}$, let $\Bhat_t(D)$ be the largest integer $b\ge0$ such that $|S|^{-1}\sum_{(i,j)\in S}\Phi^{ij}_t(b)\ge1/\alpha$ for every $S\in\mathcal{S}$ with $S\cap D\neq\emptyset$ ($-1$ if there is none). Against attackers in $\classP$ (with $\cins$) or $\classV$ (with $\crep$), $\P\big(\exists\,t,\ \exists\text{ false joint claim }D:\ \Bhat_t(D)\ge N_t\big)\le\alpha$.
\end{proposition}

Every published comparison, rank confidence set or set of the $r$ best models can thus carry one tolerance, valid over all claims and times against one attacker. This is closed testing \citep{marcus1976} on deflated wealths, averaged as in anytime-valid rank sets \citep{khosravi2026rank}. With $\mathcal{S}$ all sets of pairs it needs no assumption on which claims are correct.

\subsection{The lead claim}\label{sec:average}

\begin{theorem}[Validity for the lead claim]\label{thm:average}
Let $\Bhat^Z_t$ be the lead tolerance of \cref{sec:expwealth} with constants $4>\gamma_1>\dots>\gamma_m>0$ and weights $w_k>0$, $\sum_kw_k\le1$.
(a) With the cost $\gamma/2-\gamma^2/8$, $\Bhat^Z_t$ is $\alpha$-valid for the lead claim against the attackers in $\classP$ that only insert, and (b) with the cost $\gamma$ against $\classV$.
(c) Against $\classO$ the same holds with the cost $\gamma$ and the margin of the uncorrupted sequence.
(d) In (a) to (c), $\Bhat^Z_t$ is also $\alpha$-valid in the sense of \eqref{eq:validity}. For i.i.d.\ uncorrupted values with mean $\frac12+\delta$, ties allowed, and no attack, almost surely $\Bhat^Z_t/t\to\rho^Z=\max_k(\gamma_k\delta-\gamma_k^2/8)^+/c^Z(\gamma_k)$, with $\rho^Z\ge2\delta-\gamma_m/4$ (insertion cost) and $\rho^Z\ge\delta-\gamma_m/8$ (replacement cost).
\end{theorem}

\Cref{thm:average} gives the lead claim the same guarantee as the edge claim, at nearly the same costs and the same optimal rates. The difference is in part (a). A record overwritten unseen still enters $\Delta_t$ through its hidden value, so part (a) covers insertions only, unlike \cref{thm:validity}(a) for the edge claim. A tie costs the exponential wealth a little and the betting wealth nothing, so the lead tolerance falls slightly below the edge tolerance under frequent ties. Simultaneous lead certificates follow by a Bonferroni correction (\cref{cor:avg-simul}). Proofs in \cref{app:proofs}.

\section{Experiments}\label{sec:experiments}

We check the results by simulation and by replaying rigging attacks on real Arena votes and benchmark items. The setup is in \cref{app:overview}.

\textbf{Methods.} Ours is the grid tolerance of \eqref{eq:bhat-grid}, run as the \emph{insertion certificate} for forged records and as the \emph{replacement certificate} for flipped or contaminated ones. The baselines are \emph{undeflated betting}, a \emph{per-step robust} test of the kind in \cref{prop:perstep}, and a \emph{Wald interval} recomputed as records arrive.

\subsection{Validity against three attacks}

This experiment tests whether each method keeps its error at $\alpha$ under the attacks of \cref{thm:validity,thm:sharp}. \Cref{fig:validity} shows that only the replacement certificate stays valid against every attacker and that only flips defeat the insertion certificate, as the theorems predict. Under a burst of forged wins (panel a), undeflated betting and the per-step robust test err with probability $0.33$ and $0.27$ at $B=20$ and always for $B\ge40$, while our insertion certificate stays at or below $0.048$. With losses flipped (panel b), the insertion certificate fails from $B=5$ (at $B=20$, $0.37$ with a single bet and $0.94$ with the grid) and the replacement certificate stays at or below $0.050$. With forged wins spread at random (panel c) the per-step robust test is valid. Standard errors are $\le0.005$.

\subsection{Growth of the tolerance}

This experiment measures how fast the tolerance grows against the rates of \cref{sec:rates}. \Cref{fig:rates}(a,b) shows that the tolerance per record rises toward the optimal rates $2\delta$ and $\delta$ and that the upper bound on the time to certify is nearly attained. At $\delta=0.1$ and $t=8{,}000$ the median $\Bhat_t/t$ is $0.17$ (insertion) and $0.083$ (replacement), against $0.2$ and $0.1$. For $b\ge50$ the mean $\tau_b$ is at most $19\%$ below \eqref{eq:tau-upper}; at $b=400$ it is at most $1.9$ times the lower bound on the mean time to certify (\cref{cor:lower-mean}, \cref{app:proofs}) for $\delta=0.05$ and replacement needs $1.6$--$1.9$ times the records of insertion.

\subsection{Chatbot Arena and static benchmarks}

We test the certificates on real data because that is where they would be used, with uneven pair counts, frequent ties and attacks on published standings.

\textbf{Arena votes.} We certify pairwise claims on the time-ordered votes of the public Arena file ($1.8$ million votes, $129$ models, April 2023 to August 2024, $36\%$ ties; \citealp{chiang2024arena,chatbotarena_data}). The certificates need a battle's pair fixed before its vote, which every sampling rule satisfies. \Cref{fig:rates}(c) shows four pairs over time. Well-separated models are certified against about $1{,}300$ to $2{,}000$ forged votes and half as many flipped ones, the closest pair about $200$. For separated pairs the completed vote set withstands only slightly more, and the lead certificate slightly less (\cref{tab:arena-pairs} in \cref{app:experiments}). When all pairs are claimed at once, no certificate separates GPT-4o (2024-05-13) and \mbox{Claude-3.5-Sonnet} (\cref{fig:matrix} in \cref{app:experiments}).

\begin{table}[t]
\caption{Arena rigging replay. Probability of certifying a false ordering at $B=0$ / $100$ / $400$ attacker votes (mean of three pairs, $1{,}000$ replays each; valid if at most $0.05$). Burst: forged votes at one random time. Flip: the weaker model's losses flipped.}
\label{tab:rigging}
\centering
\setlength{\tabcolsep}{3pt}
\begin{tabular}{lcc}
\toprule
Method & Burst & Flip \\
\midrule
Wald interval & .07 / .40 / 1.0 & .07 / 1.0 / 1.0 \\
Undeflated betting & .00 / .22 / .99 & .00 / 1.0 / 1.0 \\
Per-step robust & .01 / .13 / .28 & .01 / 1.0 / 1.0 \\
Ours, insertion & .00 / .01 / .01 & .00 / 1.0 / 1.0 \\
Ours, replacement & .00 / .01 / .01 & .00 / .01 / .01 \\
\bottomrule
\end{tabular}
\end{table}

\textbf{Rigging replay.} \Cref{tab:rigging} asks whether a few hundred rigged votes make a method certify a false ordering. For three pairs whose promoted model is slightly worse, an attacker adds or flips votes in resampled replays. Every baseline breaks at a few hundred votes, while the matching certificate stays valid.

\textbf{Static benchmarks.} Benchmarks are the setting of the offline attacker. A developer who knows every test item can contaminate the ones it chooses, so the replacement certificate is required. On Massive Multitask Language Understanding (MMLU; $14{,}042$ questions, four open-weight 7--9B models; \citealp{liang2023helm,helm_data,hendrycks2021mmlu}) well-separated pairs withstand about nine tenths of the replaced items the released benchmark could withstand at most (\cref{prop:finite}). On five frontier benchmarks (MMLU-Pro, GPQA, IFEval, Omni-MATH and SWE-bench Verified; \citealp{wang2024mmlupro,rein2023gpqa,zhou2023ifeval,gao2025omnimath,jimenez2024swebench,openai2024swebenchverified,swebench_experiments}) no top-two ordering is certified when all pairs are claimed at once, and no top-five claim withstands over $7$ contaminated items (\cref{fig:frontier} in \cref{app:experiments}).

\textbf{Additional experiments.} \Cref{app:experiments} also reports the price of selecting among variants (E4), a regime shift (E5), attacks on correct claims (E6), the rigging replay on benchmark items (E7), the lead claim under ties (E9) and fixed-time breakdown counts (E10).

\section{Conclusion and Discussion}\label{sec:conclusion}

Leaderboards face adaptive attackers, and existing guarantees assume genuine records or a per-step corruption level. We introduce corruption tolerance, the number of corrupted records a ranking claim withstands, and certify it at all times against attackers who watch the published tolerances and corrupt in bursts. It charges about twice as much for an altered as for a forged record and grows nearly as fast as possible on i.i.d.\ records without ties. A second certificate covers the claim that a model leads on average so far, and selection among $V$ variants costs about $\log V$ nats. On Arena it certifies some orderings against about $2{,}000$ forged votes but not GPT-4o over Claude-3.5-Sonnet claimed at once (limitations: \cref{app:limitations}).

\textbf{Reproducibility.} Details are in \cref{app:repro}. Code, data and saved results are at \url{https://github.com/HamedKhosravi99/corruption-tolerance}.

\onecolumn
\appendix
\crefalias{section}{appendix}
\crefalias{subsection}{appendix}
\crefname{appendix}{Appendix}{Appendices}
\renewcommand{\topfraction}{0.9}\renewcommand{\bottomfraction}{0.6}\renewcommand{\textfraction}{0.07}\renewcommand{\floatpagefraction}{0.8}\setcounter{topnumber}{3}\setcounter{totalnumber}{5}
\etocdepthtag.toc{mtappendix}
\etocsettagdepth{mtmain}{none}
\etocsettagdepth{mtappendix}{subsection}
\etocsettocstyle{\section*{Contents of the Appendix}\begingroup\setlength{\parskip}{0pt plus 1pt}}{\endgroup\clearpage}
\tableofcontents

\section{The Algorithm}\label{app:algorithm}

\begin{algorithm}[h]
\caption{Corruption tolerance for the edge claim}\label{alg:ccb}
\begin{algorithmic}[1]
\REQUIRE level $\alpha$; constant bets $\lb_1,\dots,\lb_m\in(0,2)$; weights $w_k>0$ with $\sum_k w_k\le1$; cost $c=\cins$ (attackers in $\classP$) or $c=\crep$ (attackers in $\classP$, $\classV$ or $\classO$)
\STATE $G_k\gets0$ for $k=1,\dots,m$ \COMMENT{log-wealths}
\FOR{each new record $t=1,2,\dots$ with outcome $W_t$}
\STATE $G_k\gets G_k+\log\big(1+\lb_k(W_t-\frac12)\big)$ for all $k$
\STATE $\Phi_t(b):=\sum_k w_k\exp\{G_k-b\,c(\lb_k)\}$
\STATE publish $\Bhat_t=\max\{b\in\{0,1,\dots\}:\Phi_t(b)\ge1/\alpha\}$ (bisection), or $-1$ if $\Phi_t(0)<1/\alpha$
\ENDFOR
\end{algorithmic}
\end{algorithm}

\paragraph{Why each wealth is deflated by its own cost.} Each wealth must be deflated by its own cost before averaging. Deflating the average of the wealths instead would have to use the largest cost in the grid, which would waste the small bets, the ones that reach large tolerances.

\section{Proofs}\label{app:proofs}

\subsection{Formal model and two facts used throughout}\label{app:model}

\paragraph{Filtrations.} All random objects live on one probability space. The attacker's private randomization is a random element $U$. For attackers in $\classP$ and $\classV$ there is a filtration $(\cG_t)_{t\ge0}$, with $U$ being $\cG_0$-measurable, that has four properties. The observed outcome $W_t$, the corruption indicator $C_t$ and, if record $t$ is not inserted, its uncorrupted value $X_t$ are $\cG_t$-measurable. The indicator $I_t$ that record $t$ is inserted is $\cG_{t-1}$-measurable. The observed $\sigma$-field $\cF_t$, generated by $W_1,\dots,W_t$, satisfies $\cF_t\subseteq\cG_t$. Finally, $W_t=X_t$ whenever $C_t=0$. On $\{I_t=1\}$ we set $X_t=\frac12$ for definiteness; every expression below that involves $X_t$ on an inserted record is multiplied by $1-I_t$ or $1-C_t$, so the choice has no effect. Writing $\1\{\cdot\}$ for the indicator of an event, the null hypothesis of \eqref{eq:null} reads $\E[(X_t-\frac12)\1\{I_t=0\}\mid\cG_{t-1}]\le0$. Every attacker that satisfies these conditions is in $\classV$. Its $C_t$ and $W_t$ may depend on $X_t$. It is in $\classP$ if, in addition, $C_t$ is $\cG_{t-1}$-measurable. This class is slightly larger than the one described in \cref{sec:attackers}, which also fixes the corrupted value in advance. The proof of \cref{thm:validity}(a) does not need that. For attackers in $\classO$ the uncorrupted values form a sequence $(X_u)_{u\ge1}$ whose law does not depend on the attacker, with $\E[X_u\mid X_1,\dots,X_{u-1}]\le\frac12$. The observed record $t$ is either inserted or holds the uncorrupted value $X_{h(t)}$, possibly replaced, where $h(t)$ is the number of non-inserted records among the first $t$. The sequence $(W_t,C_t,I_t)_{t\ge1}$ may be any measurable function of $(X_u)_{u\ge1}$ and $U$ such that $W_t=X_{h(t)}$ whenever $C_t=0$. In every class $N_t=\sum_{s\le t}C_s$.

\paragraph{Bets.} Bet $k$ uses an $\cF_{t-1}$-measurable $\lambda_{k,t}\in[0,\lb_k]$ (constant bets $\lambda_{k,t}=\lb_k$ are a special case). Then $K^{(k)}_t=\prod_{s\le t}\big(1+\lambda_{k,s}(W_s-\frac12)\big)$, and each factor $1+\lambda_{k,t}(x-\frac12)$ lies in $[1-\lb_k/2,\,1+\lb_k/2]$ for $x\in[0,1]$. Let
\[
\Phi_t(b)=\sum_{k=1}^m w_k K^{(k)}_t e^{-b\,c(\lb_k)},\qquad \Bhat_t=\max\{b\in\{0,1,\dots\}:\Phi_t(b)\ge1/\alpha\},
\]
with $\Bhat_t=-1$ if $\Phi_t(0)<1/\alpha$. This is \eqref{eq:bhat-grid}. With $m=1$ and $w_1=1$ it is \eqref{eq:bhat-single}. Because $\Phi_t$ is decreasing in $b$ and the set in the definition is finite when $\Phi_t(0)\ge1/\alpha$ (as $\Phi_t(b)\to0$), $\Bhat_t$ is well defined. The proofs compare the observed wealth with an \emph{uncorrupted wealth}, a wealth computed from uncorrupted values alone. Its form depends on the attacker class.

\begin{lemma}[Reduction to an uncorrupted wealth]\label{lem:reduction}
Suppose there are nonnegative processes $L^{(k)}$ with $K^{(k)}_t\le e^{c(\lb_k)N_t}L^{(k)}_t$ for all $k$ and $t$, pathwise. Then $\{\Bhat_t\ge N_t\}\subseteq\{\sum_k w_kL^{(k)}_t\ge1/\alpha\}$ for every $t$.
\end{lemma}

\begin{proof}
On $\{\Bhat_t\ge N_t\}$ we have $\Bhat_t\ge0$, so $\Phi_t(\Bhat_t)\ge1/\alpha$, and since $\Phi_t$ is decreasing, $\Phi_t(N_t)\ge \Phi_t(\Bhat_t)\ge1/\alpha$. Finally $\Phi_t(N_t)=\sum_kw_kK^{(k)}_te^{-c(\lb_k)N_t}\le\sum_kw_kL^{(k)}_t$.
\end{proof}

\begin{lemma}[Ville's inequality]\label{lem:ville}
If $(L_t)_{t\ge0}$ is a nonnegative supermartingale with respect to some filtration and $L_0\le1$, then $\P(\exists t:L_t\ge1/\alpha)\le\alpha$.
\end{lemma}

This is classical. \citet{ville1939} proved it for nonnegative martingales that start at one. \citet{howard2021} use it in the form above, for nonnegative supermartingales. A weighted sum $\sum_kw_kL^{(k)}_t$ of nonnegative supermartingales with respect to the same filtration, with $\sum_kw_k\le1$ and $L^{(k)}_0=1$, is again a nonnegative supermartingale, with initial value at most $1$.

\subsection{Proof of Proposition~\ref{prop:perstep}}

Under the null every $\Bern(p)$ with $p\le\frac12$ is an allowed center. The distance from the point mass at $1$ to $\Bern(p)$ is $1-p$, which is smallest, $\frac12$, at $p=\frac12$, so a forged win is allowed only if $\varepsilon\ge\frac12$. Conversely, the distance from $\Bern(r)$ to the allowed center $\Bern(\frac12)$ is $|r-\frac12|\le\frac12$ for every $r\in[0,1]$. Hence, if $\varepsilon\ge\frac12$, every conditional law of a binary outcome is allowed at every step, so every law of a binary sequence satisfies the null. A test of level $\alpha$ then rejects with probability at most $\alpha$ under each of them. \hfill$\square$

\subsection{Proof of Theorem~\ref{thm:validity}}

\emph{(a) Class $\classP$, cost $\cins$.} Define $L^{(k)}_t=\prod_{s\le t}\big(1+(1-C_s)\lambda_{k,s}(X_s-\frac12)\big)$, where the factor is $1$ when $C_s=1$. This $L^{(k)}$ is the uncorrupted wealth for $\classP$. It applies the same bets $\lambda_{k,s}$ to the uncorrupted records and skips every corrupted one. Since $C_t$ and $\lambda_{k,t}$ are $\cG_{t-1}$-measurable and $C_t=0$ implies $I_t=0$,
\[
\E[L^{(k)}_t\mid\cG_{t-1}]=L^{(k)}_{t-1}\Big(1+(1-C_t)\lambda_{k,t}\,\E\big[(X_t-\tfrac12)\1\{I_t=0\}\mid\cG_{t-1}\big]\Big)\le L^{(k)}_{t-1},
\]
using $(1-C_t)\1\{I_t=0\}=1-C_t$. So $L^{(k)}$ is a nonnegative supermartingale for $(\cG_t)$ with $L^{(k)}_0=1$. Pathwise, compare the factors of record $s$. If $C_s=0$ then $W_s=X_s$ and they agree. If $C_s=1$, the factor of $K^{(k)}$ is at most $1+\lb_k/2=e^{\cins(\lb_k)}$ and that of $L^{(k)}$ is $1$. Hence $K^{(k)}_t\le e^{\cins(\lb_k)N_t}L^{(k)}_t$. By \cref{lem:reduction,lem:ville}, ${\P(\exists t:\Bhat_t\ge N_t)}\le{\P(\exists t:\sum_kw_kL^{(k)}_t\ge1/\alpha)}\le\alpha$.

\emph{(b) Class $\classV$, cost $\crep$.} Define $L^{(k)}_t=\prod_{s\le t}\big(1+(1-I_s)\lambda_{k,s}(X_s-\frac12)\big)$, which keeps the uncorrupted value of every replaced record. Since $I_t$ and $\lambda_{k,t}$ are $\cG_{t-1}$-measurable, the same computation shows that $L^{(k)}$ is a nonnegative supermartingale for $(\cG_t)$. Pathwise, compare the factors of record $s$. If $C_s=0$ they agree. If $s$ is inserted, the ratio of the factor of $K^{(k)}$ to that of $L^{(k)}$ is at most $1+\lb_k/2\le e^{\crep(\lb_k)}$. If $s$ is replaced, the ratio is at most $(1+\lambda_{k,s}/2)/(1-\lambda_{k,s}/2)\le(1+\lb_k/2)/(1-\lb_k/2)=e^{\crep(\lb_k)}$, because $x\mapsto(1+x/2)/(1-x/2)$ is increasing on $[0,2)$. Hence $K^{(k)}_t\le e^{\crep(\lb_k)N_t}L^{(k)}_t$, and we conclude as in (a).

\emph{(c) Class $\classO$, cost $\crep$, constant bets.} Let $\tilde L^{(k)}_u=\prod_{u'\le u}\big(1+\lb_k(X_{u'}-\frac12)\big)$, a nonnegative supermartingale for the natural filtration of $(X_u)$, and set $L^{(k)}_t=\tilde L^{(k)}_{h(t)}$. Record $t$ of the observed sequence may be inserted. It then contributes a factor at most $1+\lb_k/2\le e^{\crep(\lb_k)}$ to $K^{(k)}_t$ and none to $L^{(k)}_t$. Otherwise it holds the uncorrupted value $X_{h(t)}$, and its observed factor is at most $e^{\crep(\lb_k)}$ times its uncorrupted factor if the value was replaced and equal to it otherwise. Because the bet is constant, the uncorrupted factor does not depend on where the record appears in the observed sequence. Hence $K^{(k)}_t\le e^{\crep(\lb_k)N_t}\tilde L^{(k)}_{h(t)}$ pathwise. By \cref{lem:reduction}, $\{\exists t:\Bhat_t\ge N_t\}\subseteq\{\exists u:\sum_kw_k\tilde L^{(k)}_u\ge1/\alpha\}$, an event that depends on the uncorrupted sequence alone and has probability at most $\alpha$ by \cref{lem:ville}. The argument does not use how the attacker chooses its actions. In particular $U$ may depend on $(X_u)$, and only the law of $(X_u)$ must not depend on the attacker.

\emph{Adaptive bets in (c).} With a bet computed from the observed records, the factor of the uncorrupted value $X_u$ depends on observed records that an offline attacker may have chosen after seeing $X_u$, so $\tilde L^{(k)}$ need not be a supermartingale. We do not claim validity in that case. \hfill$\square$

\subsection{Proof of Theorem~\ref{thm:sharp}}

Throughout, the bet is the constant $\lambda$, and for an integer $N\ge0$, $\Bhat_t\ge N$ if and only if $\log(\alpha K_t)\ge cN$.

\emph{(a)} The attacker inserts $B$ wins as records $1,\dots,B$. Then $K_B=(1+\lambda/2)^B=e^{B\cins}$ and $N_B=B$, so $\Bhat_B\ge N_B$ if and only if $B\cins-\log(1/\alpha)\ge Bc$, that is, $B\ge\log(1/\alpha)/(\cins-c)$. This event does not depend on the uncorrupted values.

\emph{(b)} Let $T_B$ be the position of the $B$th uncorrupted loss and $R_B=T_B-B$ the number of uncorrupted wins before it. The attacker replaces the first $B$ losses, so records $1,\dots,T_B$ are all observed as wins, $K_{T_B}=e^{T_B\cins}$ and $N_{T_B}=B$. A false statement is made at $T_B$ if $T_B\cins-\log(1/\alpha)\ge Bc$, that is, if
\[
R_B\ \ge\ B\,\frac{c-\cins}{\cins}+\frac{\log(1/\alpha)}{\cins}.
\]
Since $R_B$ is a sum of $B$ independent geometric variables with mean $1$, $R_B/B\to1$ almost surely. If $c<2\cins$, then $(c-\cins)/\cins<1$, so the probability of the display tends to one as $B\to\infty$. Since $\cins<2\cins$, this applies in particular to $c=\cins$.

\emph{(c)} The attacker replaces every uncorrupted loss by a win. Then every observed record is a win and $\log K_t=t\cins$. Also $N_t=\Lambda_t$, the number of uncorrupted losses among the first $t$ records. A false statement occurs at $t$ if and only if $S_t:=t\cins-c\Lambda_t\ge\log(1/\alpha)$. The walk $S$ has independent increments equal to $\cins$ or $\cins-c$ with probability $\frac12$ each. If $c=2\cins$, the walk has mean-zero bounded increments and is recurrent, so it reaches $\log(1/\alpha)$ with probability one. If $2\cins<c<\crep$, let $\varphi(\theta)=\frac12(e^{\theta\cins}+e^{\theta(\cins-c)})$. It is convex with $\varphi(0)=1$ and $\varphi'(0)=\cins-c/2<0$, and $\varphi(1)=\frac12\Mp(1+e^{-c})>\frac12\Mp(1+\mm/\Mp)=1$ because $e^{-c}>e^{-\crep}=\mm/\Mp$. So there is a unique $\theta\in(0,1)$ with $\varphi(\theta)=1$, and $e^{\theta S_t}$ is a martingale. The drift of $S$ is negative, so $S_t\to-\infty$ almost surely. Let $\sigma=\inf\{t:S_t\ge\log(1/\alpha)\}$ and write ${t\wedge\sigma=\min(t,\sigma)}$. Since $S_{t\wedge\sigma}<\log(1/\alpha)+\cins$, bounded convergence gives $1=\E[e^{\theta S_\sigma}\1\{\sigma<\infty\}]\le e^{\theta(\log(1/\alpha)+\cins)}\P(\sigma<\infty)$, so $\P(\sigma<\infty)\ge\alpha^{\theta}e^{-\theta\cins}$. This exceeds $\alpha$ if and only if $\alpha<e^{-\theta\cins/(1-\theta)}$.

\emph{(d)} Consider any attacker in $\classO$ that makes no insertions and at most $B$ replacements, and fix the uncorrupted sequence. Because the bet is constant and nothing is inserted, replacing record $s$ multiplies $K_t$, for every $t\ge s$, by the ratio of the new factor to the uncorrupted one. Whatever value in $[0,1]$ the attacker writes, this ratio is at most $e^{\crep}$ if the uncorrupted value was a loss and at most $1$ if it was a win. Let $\Lambda_t$ be the number of uncorrupted losses among the first $t$ records, and $A_t$ the number of these whose value was replaced. Then $\log K_t\le\log\tilde L_t+\crep A_t$, $A_t\le\min(B,\Lambda_t)$ and $N_t\ge A_t$, where $\tilde L_t=\prod_{s\le t}(1+\lambda(X_s-\frac12))$ is the wealth of the uncorrupted sequence, as in the proof of \cref{thm:validity}(c). The attacker of (b) achieves $A_t=N_t=\min(B,\Lambda_t)$ and $\log K_t=\log\tilde L_t+\crep\min(B,\Lambda_t)$ for all $t$, which is the maximum of $K_t$. Moreover, for $c\le\crep$,
\[
\log(\alpha K_t)-cN_t\le\log(\alpha\tilde L_t)+(\crep-c)A_t\le\log(\alpha\tilde L_t)+(\crep-c)\min(B,\Lambda_t),
\]
and the right-hand side is the value of $\log(\alpha K_t)-cN_t$ for the attacker of (b). So whenever any such attacker makes a false statement at time $t$, the attacker of (b) does too, on the same uncorrupted sequence. For the grid certificate, the same computation applies to each bet. The attacker of (b) maximizes every $K^{(k)}_t$ at once. For each $k$ the computation gives $w_kK^{(k)}_te^{-c_kN_t}\le w_k\tilde L^{(k)}_te^{(\crep(\lb_k)-c_k)\min(B,\Lambda_t)}$ because $c_k\le\crep(\lb_k)$, with equality for the attacker of (b). Summing over $k$ shows that $\Phi_t(N_t)$ is largest under the attacker of (b). Since $N_t$ is a nonnegative integer, $\Phi_t(N_t)\ge1/\alpha$ exactly when $\Bhat_t\ge N_t$. So the attacker of (b) again makes a false statement whenever any such attacker does. \hfill$\square$

\subsection{Proof of Theorem~\ref{thm:upper}}

\emph{The growth rate.} With $\P(W=1)=a$, $\P(W=0)=\ell$ and $\P(W=\frac12)=1-a-\ell$,
\[
g(\lambda)=a\log(1+\tfrac\lambda2)+\ell\log(1-\tfrac\lambda2)
=\tfrac{a+\ell}{2}\log(1-\tfrac{\lambda^2}{4})+\tfrac{a-\ell}{2}\log\tfrac{1+\lambda/2}{1-\lambda/2}
=\delta\,\crep(\lambda)-\tfrac\pi2\kappa(\lambda),
\]
which is \eqref{eq:growth}.

\emph{The limit of \eqref{eq:rate}.} By the strong law of large numbers, $t^{-1}\log K^{(k)}_t\to g(\lb_k)$ almost surely for each $k$. Two bounds sandwich $\Bhat_t$. First, $\Phi_t(b)\ge w_kK^{(k)}_te^{-bc(\lb_k)}$, so $\Bhat_t\ge\max_k\lfloor(\log K^{(k)}_t+\log(\alpha w_k))/c(\lb_k)\rfloor$. Second, $\Phi_t(b)\le\max_kK^{(k)}_te^{-bc(\lb_k)}$ because $\sum_kw_k\le1$, so $\Bhat_t\le\max\{-1,\max_k\lfloor(\log K^{(k)}_t+\log\alpha)/c(\lb_k)\rfloor\}$. Dividing by $t$, the lower bound converges to $\max_kg(\lb_k)/c(\lb_k)$ and the upper bound to $\max\{0,\max_kg(\lb_k)/c(\lb_k)\}=\rho$. Together with $\Bhat_t\ge-1$, this gives $\Bhat_t/t\to\rho$.

\emph{The bound of \eqref{eq:tau-upper}.} Fix $k$ with $g(\lb_k)>0$ and let $z_k=\log\frac1{\alpha w_k}+b\,c(\lb_k)$ and $\sigma_k={\inf\{t:\log K^{(k)}_t\ge z_k\}}$. At $\sigma_k$, $\Phi_{\sigma_k}(b)\ge w_kK^{(k)}_{\sigma_k}e^{-bc(\lb_k)}\ge1/\alpha$, so $\tau_b\le\sigma_k$. (For an integer $b$, $\Phi_t(b)\ge1/\alpha$ is equivalent to $\Bhat_t\ge b$.) The log-wealth $\log K^{(k)}_t$ is a random walk with independent, bounded increments of mean $g(\lb_k)>0$, bounded above by $\cins(\lb_k)$. By Wald's identity at the bounded stopping time $\sigma_k\wedge n$, $g(\lb_k)\E(\sigma_k\wedge n)=\E\log K^{(k)}_{\sigma_k\wedge n}\le z_k+\cins(\lb_k)$, since before $\sigma_k$ the walk is below $z_k$ and it crosses $z_k$ by at most one increment. Letting $n\to\infty$ (monotone convergence) gives $\E\,\sigma_k\le(z_k+\cins(\lb_k))/g(\lb_k)$. Minimizing over $k$ gives \eqref{eq:tau-upper}.

\begin{corollary}[Explicit constants]\label{cor:explicit}
Under the conditions of \cref{thm:upper}, if $\delta\le\frac13$ and the grid contains a bet $\lb_k\in[3\delta/\sqrt2,3\delta]$, then for every integer $b\ge0$,
$\E\,\tau_b\le\delta^{-2}\log\frac{1}{\alpha w_k}+1.5(b+1)/\delta$ with the insertion cost and
$\E\,\tau_b\le\delta^{-2}\log\frac{1}{\alpha w_k}+(4b+1.5)/\delta$ with the replacement cost.
\end{corollary}

\emph{Explicit bounds, and proof of \cref{cor:explicit}.} For $\lambda\in(0,1]$, using $\crep(\lambda)=2\,\mathrm{artanh}(\lambda/2)$, $x\le\mathrm{artanh}(x)\le x/(1-x^2)$ and $x/(1+x)\le\log(1+x)\le x$,
\[
\lambda\le\crep(\lambda)\le\frac{\lambda}{1-\lambda^2/4}\le\frac43\lambda,\qquad
\frac\lambda3\le\frac{\lambda/2}{1+\lambda/2}\le\cins(\lambda)\le\frac\lambda2,\qquad
\kappa(\lambda)\le\frac{\lambda^2/4}{1-\lambda^2/4}\le\frac{\lambda^2}{3}.
\]
Using $\crep=2\cins+\kappa\ge2\cins$,
\[
\frac{g}{\cins}\ge2\delta-\frac\pi2\cdot\frac{\kappa}{\cins}\ge2\delta-\frac{\pi\lambda}{2},\qquad
\frac{g}{\crep}=\delta-\frac\pi2\cdot\frac{\kappa}{\crep}\ge\delta-\frac{\pi\lambda}{6}.
\]
Also $g(\lambda)\ge\delta\lambda-\frac\pi2\cdot\frac{\lambda^2}{3}\ge\lambda(\delta-\lambda/6)$. If $\lb_k\in[3\delta/\sqrt2,3\delta]$ and $3\delta\le1$, then $\delta-\lb_k/6\ge\delta/2$, so $g(\lb_k)\ge\lb_k\delta/2\ge\frac{3}{2\sqrt2}\delta^2\ge1.06\,\delta^2$, while $\cins(\lb_k)\le1.5\delta$ and $\crep(\lb_k)\le4\delta$. Substituting into \eqref{eq:tau-upper} gives $\E\,\tau_b\le(\log\frac1{\alpha w_k}+1.5\delta b+1.5\delta)/(1.06\delta^2)$ with the insertion cost and $\E\,\tau_b\le(\log\frac1{\alpha w_k}+4\delta b+1.5\delta)/(1.06\delta^2)$ with the replacement cost, which are at most the stated bounds. \hfill$\square$

\subsection{Proof of Theorem~\ref{thm:lower} and of the optimal rate}\label{app:lower}

Let $E=\{\max_{t\le n}\Bhat_t\ge b\}$. By \cref{def:budget} it is an event determined by the first $n$ observed records, and independent auxiliary randomization in $\Bhat_t$ would change nothing below.

\emph{The coupling bound of \eqref{eq:lower-coupling}, insertion.} Let $\zeta_1,\zeta_2,\dots$ be independent $\Bern(2\delta)$ and $X'_1,X'_2,\dots$ independent $\Bern(\frac12)$, all independent. Define $Y_t=\zeta_t+(1-\zeta_t)X'_t$. Then $Y_1,Y_2,\dots$ are independent $\Bern(\frac12+\delta)$, since $\P(Y_t=1)=2\delta+(1-2\delta)/2$. Write $\P_\delta$ and $\E_\delta$ for probability and expectation when the records are independent $\Bern(\frac12+\delta)$ and there is no attack. The case $\delta=0$ is $\P_0$. Now define an attacker in $\classP$ in a null world, that is, one in which $H_0^{ij}$ holds. At position $t$, if $\zeta_t=1$ and fewer than $\lfloor b\rfloor$ insertions have been made, it inserts a win. Otherwise the next uncorrupted record appears, and we let its value be $X'_t$. The uncorrupted records are then independent $\Bern(\frac12)$, because the positions at which they appear are determined by $\zeta$, which is independent of $X'$. So $H_0^{ij}$ holds. The attacker decides before each record using only its own coins, so it is in $\classP$, and it never corrupts more than $\lfloor b\rfloor$ records. On the event $\{\sum_{t\le n}\zeta_t\le\lfloor b\rfloor\}$, the first $n$ observed records coincide with $Y_1,\dots,Y_n$. Therefore
\[
\P_\delta(E)=\P\big((Y_1,\dots,Y_n)\in E\big)\le\P\big((W_1,\dots,W_n)\in E\big)+\P\big(\textstyle\sum_{t\le n}\zeta_t>\lfloor b\rfloor\big)\le\alpha+\P\big(\Bin(n,2\delta)>b\big),
\]
where the second inequality uses that on $E$ we have $\Bhat_t\ge b\ge N_t$ for some $t\le n$, which has probability at most $\alpha$ by validity, and that $\sum_{t\le n}\zeta_t\sim\Bin(n,2\delta)$ is an integer, so it exceeds $\lfloor b\rfloor$ exactly when it exceeds $b$.

\emph{Replacement.} Take $\zeta_t\sim\Bern(2\delta)$ and $X'_t\sim\Bern(\frac12)$ as before and $Y_t=X'_t+(1-X'_t)\zeta_t$, again independent $\Bern(\frac12+\delta)$. The attacker observes each uncorrupted value $X'_t$ and, if $X'_t=0$, $\zeta_t=1$ and fewer than $\lfloor b\rfloor$ replacements have been made, replaces it by a win. This attacker is in $\classV$. Without the limit of $\lfloor b\rfloor$ replacements, the attacker would make $\sum_{t\le n}(1-X'_t)\zeta_t\sim\Bin(n,\delta)$ replacements among the first $n$ records, and the same argument gives \eqref{eq:lower-coupling} with $q=\delta$.

\emph{The information bound of \eqref{eq:lower-kl}.} Let $E_0=\{\max_{t\le n}\Bhat_t\ge0\}$, an event determined by the first $n$ records. With no attack, $N_t=0$, so validity gives $x_0:=\P_0(E_0)\le\alpha$. Let $x=\P_\delta(E_0)$. By the data-processing inequality for the laws of the first $n$ records under $\P_\delta$ and $\P_0$, $\mathrm{kl}(x,x_0)\le n\KL_\delta$, where $\mathrm{kl}(p,r)=p\log\frac pr+(1-p)\log\frac{1-p}{1-r}$. If $x\ge\alpha$, then $\mathrm{kl}(x,x_0)\ge\mathrm{kl}(x,\alpha)\ge x\log\frac1\alpha-\log2$, which gives $x\le(n\KL_\delta+\log2)/\log(1/\alpha)$. If $x<\alpha$, the bound holds because $\log2/\log(1/\alpha)\ge\alpha$ for all $\alpha\in(0,1)$ (as $\alpha\log(1/\alpha)\le e^{-1}<\log2$). Finally, $\KL_\delta$ is at most the $\chi^2$-divergence, $\chi^2=\frac{\delta^2}{1/2}+\frac{\delta^2}{1/2}=4\delta^2$.

\emph{Consequences.} For integer $b$, Markov's inequality gives $\P(\Bin(n,q)>b)=\P(\Bin(n,q)\ge b+1)\le nq/(b+1)$, so $\P_\delta(E)\ge1-\beta$ forces $nq/(b+1)\ge1-\alpha-\beta$. Since $E\subseteq E_0$ for $b\ge0$, \eqref{eq:lower-kl} forces $4n\delta^2+\log2\ge(1-\beta)\log(1/\alpha)$. By Hoeffding's inequality, $\P(\Bin(n,q)>(q+\eta)n)\le e^{-2n\eta^2}$, and $\{\Bhat_n\ge(q+\eta)n\}\subseteq\{\max_{t\le n}\Bhat_t\ge(q+\eta)n\}$. Hence \eqref{eq:lower-coupling} with $b=(q+\eta)n$ gives ${\P(\Bhat_n\ge(q+\eta)n)}\le\alpha+e^{-2n\eta^2}$.

\begin{corollary}[Lower bound on the mean time to certify]\label{cor:lower-mean}
Under the conditions of \cref{thm:lower}, for every integer $b\ge0$,
\[
\E\,\tau_b\ \ge\ \max\Big\{\frac{\log(1/\alpha)}{\KL_\delta},\ \sup_{n\ge1}\,n\big(1-\alpha-\P(\Bin(n,q)>b)\big)\Big\}.
\]
\end{corollary}

\begin{proof}
The second term follows from $\E_\delta\,\tau_b\ge n\P_\delta(\tau_b>n)$ and \eqref{eq:lower-coupling}. For the first, $\tau_b\ge\tau_0$ because $b\ge0$, and we may assume $\E_\delta\,\tau_0<\infty$. With no attack, $\tau_0$ is the rejection time of a test of $\Bern(\frac12)$ with $\P_0(\tau_0<\infty)\le\alpha$. For each $n$, $\tau_0\wedge n$ is a bounded stopping time, so Wald's identity for the log-likelihood ratio and the data-processing inequality applied to the event $\{\tau_0\le n\}\in\cF_{\tau_0\wedge n}$ give $\KL_\delta\,\E_\delta(\tau_0\wedge n)\ge\mathrm{kl}\big(\P_\delta(\tau_0\le n),\P_0(\tau_0\le n)\big)$ \citep[Lemma~1]{kaufmann2016}. As $n\to\infty$, $\P_\delta(\tau_0\le n)\to1$ and $\P_0(\tau_0\le n)\le\alpha$, so the limit inferior of the right-hand side is at least $\mathrm{kl}(1,\alpha)=\log(1/\alpha)$ (using that $\mathrm{kl}(p,r)$ is decreasing in $r\le p$), and the left-hand side increases to $\KL_\delta\,\E_\delta\,\tau_0$.
\end{proof}

\emph{Optimal rate.} With $\pi=1$ and the smallest bet $\lb_m\le1$, \cref{thm:upper} gives $\rho\ge g(\lb_m)/\cins(\lb_m)\ge2\delta-\lb_m/2$ for the insertion cost and $\rho\ge\delta-\lb_m/6$ for the replacement cost. The upper limits $\rho\le q$ follow from \cref{thm:lower}, because the tolerance is valid (\cref{thm:validity}). Indeed, if $\rho>q$, take $\eta\in(0,\rho-q)$ and $b=\lceil(q+\eta)n\rceil$. Since $\Bhat_n/n\to\rho$ almost surely, the left-hand side of \eqref{eq:lower-coupling} is at least $\P(\Bhat_n\ge b)\to1$. By Hoeffding's inequality its right-hand side is at most $\alpha+e^{-2n\eta^2}\to\alpha<1$, a contradiction. The coupling bound of \eqref{eq:lower-coupling} also holds with ties, and with it the limits $\rho\le q$, when $W$ is $1$, $0$ or $\frac12$ with probabilities $\pi/2+\delta$, $\pi/2-\delta$ and $1-\pi$. For insertion, let $X'_t$ have the null law that is $1$ or $0$ with probability $(\pi/2-\delta)/(1-2\delta)$ each and $\frac12$ otherwise, and keep $\zeta_t\sim\Bern(2\delta)$. Then $\zeta_t+(1-\zeta_t)X'_t$ has the law of $W$. For replacement, let $X'_t$ be $1$ or $0$ with probability $\pi/2$ each and $\frac12$ otherwise, and let the attacker replace a loss by a win when $\zeta_t=1$, now with $\zeta_t\sim\Bern(2\delta/\pi)$. Without the cap of $\lfloor b\rfloor$ replacements their number is $\Bin(n,\delta)$, and the observed law is that of $W$. For the insertion cost this can also be checked directly, since $g/\cins\le2\delta$ is equivalent to $\delta\kappa\le\pi\kappa/2$, which holds because $\ell\ge0$. \hfill$\square$

\subsection{Proof of Proposition~\ref{prop:selection}}

\emph{Validity.} Let $\mathcal N$ be the set of null variants. For $v\in\mathcal N$, let $K^{(v,k)}$ be the wealth of variant $v$ with bet $k$ and $N^{(v)}_t$ the number of corrupted records of variant $v$. Let $L^{(v,k)}$ be the uncorrupted wealth from the proof of \cref{thm:validity}(a) or (b), with factor $1$ on records of other variants, so that $K^{(v,k)}_t\le e^{c(\lb_k)N^{(v)}_t}L^{(v,k)}_t$. The variant of each record is fixed before its value is produced, and the null of $v$ holds given all records. So each $L^{(v,k)}$ is a nonnegative supermartingale for the global filtration, that of all records of all variants. A statement about the published variant is false at time $t$ only if $\hat v_t\in\mathcal N$ and $\sum_kw_kK^{(\hat v_t,k)}_te^{-c(\lb_k)N^{(\hat v_t)}_t}\ge V/\alpha$ (as in \cref{lem:reduction}), which implies $\bar L_t:=\frac1V\sum_{v\in\mathcal N}\sum_kw_kL^{(v,k)}_t\ge\frac1V\sum_kw_kL^{(\hat v_t,k)}_t\ge1/\alpha$. The process $\bar L$ is a nonnegative supermartingale with $\bar L_0\le|\mathcal N|/V\le1$, so \cref{lem:ville} applies. The choice of $\hat v_t$ plays no role.

\emph{Adaptive creation.} The variants need not exist from the start. Fix $V$ in advance and let the developer create and retire variants as results come in, with at most $V$ of them ever receiving a record, indexed by order of creation. Let $A_r$ be the event that the $r$th created variant is one of the null variants in $\mathcal N$. Its identity is known at creation, before its first record, so $A_r$ and, at each of its records, the variant of the record are fixed before the value is produced. Define $L^{(r,k)}$ as the uncorrupted wealth of the $r$th created variant on $A_r$ and as the constant $1$ on the complement; it equals $1$ until the variant receives its first record. On $A_r$ the null of that variant holds given the records of all variants, so each factor of $L^{(r,k)}$ has conditional expectation at most $1$ given the past. Each $L^{(r,k)}$ is then a nonnegative supermartingale for the global filtration, $\bar L_t=\frac1V\sum_{r\le V}\sum_kw_kL^{(r,k)}_t$ has $\bar L_0\le1$, and the rest of the argument is unchanged. The bound $V$ must be fixed in advance.

\emph{Necessity.} Let $K$ be the wealth of one variant with constant bet $\lambda$ and independent $\Bern(\frac12)$ records, and $y>0$. Then $K$ is a martingale, and $\log K_t$ is a random walk with mean increment $\frac12\log(1-\lambda^2/4)<0$, so $K_t\to0$ almost surely. Let $\sigma=\inf\{t\ge1:K_t\ge y\}$ and $\bar y=\max(y,1)$. Since $K_{t\wedge\sigma}\le\bar y\Mp$, bounded convergence gives $1=\E[K_\sigma\1\{\sigma<\infty\}]\le\bar y\Mp\,\P(\sigma<\infty)$, so $\P(\sup_{t\ge1}K_t\ge y)\ge1/(\bar y\Mp)$. Take $y=V'/\alpha$. If $V'<\alpha$, then $\bar y=1$ and a single variant already reaches the threshold with probability at least $1/\Mp>\frac12\ge\alpha$, so assume $V'\ge\alpha$. Publishing the variant whose wealth first reaches $V'/\alpha$ makes a false statement if some variant reaches it. With $V$ independent variants this has probability at least $1-(1-\alpha/(V'\Mp))^V\ge1-e^{-V\alpha/(V'\Mp)}$. If $V'\le V/(2\Mp)$, this is at least $1-e^{-2\alpha}>\alpha$ for $\alpha\le\frac12$. \hfill$\square$

\subsection{Proof of Proposition~\ref{prop:closed}}

Let the time $t$ run over all records on all pairs, and let $(\cG_t)$ be the global filtration. For each ordered pair $(i,j)$ and bet $k$, let $K^{(ij,k)}$ be the observed wealth of bet $k$ on the records of $(i,j)$. Define the uncorrupted wealth $L^{(ij,k)}$ as in the proof of \cref{thm:validity}(a) or (b), with factor $1$ on records that do not concern $(i,j)$. By the two assumptions of the proposition, each is a nonnegative supermartingale for $(\cG_t)$ when $H_0^{ij}$ holds. As in the proof of \cref{thm:validity}, $K^{(ij,k)}_t\le e^{c(\lb_k)N^{ij}_t}L^{(ij,k)}_t\le e^{c(\lb_k)N_t}L^{(ij,k)}_t$ pathwise, where $N^{ij}_t\le N_t$ counts corrupted records on $(i,j)$. The set $S^\ast$ is a property of the law of the data, not random. If $S^\ast=\emptyset$, no joint claim is false. Otherwise, suppose a joint claim $D$ is false at time $t$ and $\Bhat_t(D)\ge N_t$. Then $S^\ast\cap D\neq\emptyset$ and $S^\ast\in\mathcal S$, so $|S^\ast|^{-1}\sum_{(i,j)\in S^\ast}\Phi^{ij}_t(\Bhat_t(D))\ge1/\alpha$. Each $\Phi^{ij}_t$ is decreasing, so the same holds with $N_t$ in place of $\Bhat_t(D)$, and therefore $\bar L_t:=|S^\ast|^{-1}\sum_{(i,j)\in S^\ast}\sum_kw_kL^{(ij,k)}_t\ge1/\alpha$. The process $\bar L$ is a nonnegative supermartingale for $(\cG_t)$ with $\bar L_0\le1$, and it does not depend on $D$. \Cref{lem:ville} bounds the probability that it ever reaches $1/\alpha$ by $\alpha$. \hfill$\square$

On a static benchmark the records do not arrive one pair at a time. Each item scores every pair at once, and the records of one item on two pairs that share a model are dependent. The following corollary covers this case. It conditions only on earlier items, so it does not order the pairs within an item, and it counts a corrupted item once however many pairs it changes. Let $\cH_{t-1}$ be the $\sigma$-field generated by the attacker's randomization, by all variables of the first $t-1$ items (observed outcomes, indicators and uncorrupted values, as for $\cG_{t-1}$) and by $I_t$.

\begin{corollary}[Simultaneous certificates with item-level records]\label{cor:items}
Suppose that at each time $t$ one item arrives and produces a record $W^{ij}_t$, with uncorrupted value $X^{ij}_t$, for every ordered pair $(i,j)$ of a set of models fixed in advance. Let $I_t$ indicate that the item is inserted, let $C_t=1$ if it is inserted or any of its records is replaced, and let $N_t=\sum_{s\le t}C_s$ count the corrupted items. Suppose that $\E[X^{ij}_t\mid\cH_{t-1}]\le\frac12$ on $\{I_t=0\}$ for every pair $(i,j)\in S^\ast$. Then the conclusion of \cref{prop:closed} holds with this $N_t$, with the cost $\cins$ if every $C_t$ is $\cH_{t-1}$-measurable and with the cost $\crep$ in general.
\end{corollary}

\begin{proof}
Fix a bet $k$ and a pair $(i,j)\in S^\ast$. Define the uncorrupted wealth $L^{(ij,k)}$ from the records of $(i,j)$ as in the proof of \cref{thm:validity}(a) for the cost $\cins$ and (b) for the cost $\crep$, with the item-level variables $C_t$ and $I_t$. The variables $I_t$ and $\lambda_{k,t}$ are $\cH_{t-1}$-measurable, and so is $C_t$ in the first case, and the null of $(i,j)$ holds given $\cH_{t-1}$. The computation in that proof therefore shows that $L^{(ij,k)}$ is a nonnegative supermartingale for $(\cH_t)$, and this holds for every pair in $S^\ast$ with respect to the same filtration. Pathwise, at an item with $C_t=0$ the factors of $K^{(ij,k)}$ and $L^{(ij,k)}$ agree for every pair. At an item with $C_t=1$, the factor of $K^{(ij,k)}$ is at most $e^{c(\lb_k)}$ times that of $L^{(ij,k)}$. With the cost $\cins$ this holds because $L^{(ij,k)}$ skips the whole item, so its factor is $1$, while every factor of $K^{(ij,k)}$ is at most $1+\lb_k/2$. With the cost $\crep$ it is the factor comparison in the proof of \cref{thm:validity}(b), and the two factors agree when the record of $(i,j)$ itself is uncorrupted. Hence $K^{(ij,k)}_t\le e^{c(\lb_k)N_t}L^{(ij,k)}_t$ for every pair in $S^\ast$, with the same $N_t$. The rest of the proof of \cref{prop:closed} applies without change, with $(\cH_t)$ in place of $(\cG_t)$.
\end{proof}

\emph{Offline attackers.} Let the uncorrupted items form a sequence of vectors $\mathbf X_u=(X^{ij}_u)_{(i,j)}$, $u\ge1$, whose law does not depend on the attacker, and suppose that $\E[X^{ij}_u\mid\mathbf X_1,\dots,\mathbf X_{u-1}]\le\frac12$ for every $u$ and every $(i,j)\in S^\ast$. With constant bets and the cost $\crep$, the conclusion of \cref{prop:closed} then also holds against attackers in $\classO$, who may inspect the whole uncorrupted sequence before inserting items or replacing records, with $N_t$ counting the corrupted items. For the proof, let $\tilde L^{(ij,k)}_u=\prod_{u'\le u}\big(1+\lb_k(X^{ij}_{u'}-\frac12)\big)$, a nonnegative supermartingale for the filtration of the uncorrupted vectors, for every pair in $S^\ast$ at once, and let $h(t)$ be the number of non-inserted items among the first $t$ observed items. As in the proof of \cref{thm:validity}(c), an inserted item contributes a factor at most $1+\lb_k/2\le e^{\crep(\lb_k)}$ to $K^{(ij,k)}_t$ and none to $\tilde L^{(ij,k)}_{h(t)}$. A non-inserted item contributes at most $e^{\crep(\lb_k)}$ times its uncorrupted factor if the record of $(i,j)$ was replaced, and that factor otherwise. Hence $K^{(ij,k)}_t\le e^{\crep(\lb_k)N_t}\tilde L^{(ij,k)}_{h(t)}$ pathwise for every pair in $S^\ast$. The argument of \cref{prop:closed} then shows that a false joint claim $D$ with $\Bhat_t(D)\ge N_t$ forces $\sup_u|S^\ast|^{-1}\sum_{(i,j)\in S^\ast}\sum_kw_k\tilde L^{(ij,k)}_u\ge1/\alpha$, an event of the uncorrupted sequence alone, whose probability \cref{lem:ville} bounds by $\alpha$. The same holds for \cref{prop:closed} itself against $\classO$ with constant bets. It needs the uncorrupted records, each with its pair, to form a sequence whose law does not depend on the attacker, with each true null holding given the earlier uncorrupted records and the pair of the current one.

\paragraph{Computation.} Let $\mathcal S$ be the family of all sets of pairs, as in E3c and E8. For a pair and a budget $b$, the smallest average over the sets containing the pair is then attained by adding the other pairs in increasing order of their deflated averages. One check therefore sorts the pairs, and the tolerance follows by bisection over $b$.

\subsection{Proof of Theorem~\ref{thm:average} and Corollary~\ref{cor:avg-simul}}

\begin{lemma}[Hoeffding's lemma]\label{lem:hoeffding}
Let $Y$ take values in $[0,1]$, let $\cH$ be a $\sigma$-field and let $\gamma>0$. Then $\E\big[e^{\gamma(Y-\E[Y\mid\cH])}\,\big|\,\cH\big]\le e^{\gamma^2/8}$.
\end{lemma}

This is classical \citep{hoeffding1963}, applied under the conditional law given $\cH$. The betting factor $1+\lambda(W_s-\frac12)$ is not linear in the outcome, so on $\{\Delta_t\le0\}$ the betting wealth is not dominated pathwise by a supermartingale of the uncorrupted values. The exponential wealth is linear in the exponent, which the proofs below use. Throughout, $\mu_s=\E[X_s\mid\cG_{s-1}]$, which equals $\frac12$ on $\{I_s=1\}$ by the convention of \cref{app:model}. In this subsection $Z^{(k)}_t=Z^{(\gamma_k)}_t$ and $c_k=c^Z(\gamma_k)$.

\paragraph{Reduction on the null event.} \Cref{lem:reduction} uses domination at every time. Here it is needed only at times $t$ with $\Delta_t\le0$, when the lead claim is false. If $Z^{(k)}_t\le e^{c_kN_t}L^{(k)}_t$ for all $k$ on an event $\mathcal{A}_t$, then $\mathcal{A}_t\cap\{\Bhat^Z_t\ge N_t\}\subseteq\{\sum_kw_kL^{(k)}_t\ge1/\alpha\}$. The proof is that of \cref{lem:reduction}, since on $\{\Bhat^Z_t\ge N_t\}$ the number $N_t$ is a nonnegative integer.

\emph{(a) Insertion-only attackers in $\classP$, cost $\gamma/2-\gamma^2/8$.} Fix $\gamma=\gamma_k$ and write $Z_t$ for $Z^{(\gamma)}_t$. Every corrupted record is inserted, so $C_s=I_s$, and $I_s$ is $\cG_{s-1}$-measurable. Define
\[
L_t=\exp\Big(\gamma\sum_{s\le t}(1-I_s)(X_s-\mu_s)-\frac{\gamma^2}{8}\sum_{s\le t}(1-I_s)\Big).
\]
On $\{I_t=1\}$, $L_t=L_{t-1}$. On $\{I_t=0\}$, \cref{lem:hoeffding} with $\cH=\cG_{t-1}$ gives $\E[L_t\mid\cG_{t-1}]\le L_{t-1}$. Hence $L$ is a nonnegative supermartingale with $L_0=1$. Pathwise, an uncorrupted record contributes $\gamma(X_s-\frac12)-\gamma^2/8=\gamma(X_s-\mu_s)-\gamma^2/8+\gamma(\mu_s-\frac12)$ to $\log Z_t$, and an inserted record at most $\gamma/2-\gamma^2/8$. Summing,
\[
\log Z_t\le\log L_t+\gamma\Delta_t+\Big(\frac{\gamma}{2}-\frac{\gamma^2}{8}\Big)N_t .
\]
On $\{\Delta_t\le0\}$ therefore $Z^{(k)}_t\le e^{(\gamma_k/2-\gamma_k^2/8)N_t}L^{(k)}_t$ for every $k$. By the reduction above, $\{\Delta_t\le0,\ \Bhat^Z_t\ge N_t\}\subseteq\{\sum_kw_kL^{(k)}_t\ge1/\alpha\}$ for every $t$, and \cref{lem:ville} bounds the probability of the union over $t$ by $\alpha$. For every attacker in $\classP$, $C_s$ is $\cG_{s-1}$-measurable. The same computation with $C_s$ in place of $I_s$ shows that the cost $\gamma/2-\gamma^2/8$ is then valid against all of $\classP$ if $\Delta_t$ is replaced by the margin over the uncorrupted records, $\sum_{s\le t}(1-C_s)(\mu_s-\frac12)$.

\emph{(b) Class $\classV$, cost $\gamma$.} Define $L_t$ as in (a). Since $I_t$ is $\cG_{t-1}$-measurable in every class, $L$ is again a nonnegative supermartingale. Pathwise, a non-inserted record contributes $\gamma(W_s-\frac12)-\gamma^2/8$ to $\log Z_t$ and $\gamma(X_s-\mu_s)-\gamma^2/8$ to $\log L_t$. The difference is $\gamma(W_s-X_s)+\gamma(\mu_s-\frac12)\le\gamma C_s+\gamma(\mu_s-\frac12)$, since $W_s=X_s$ when $C_s=0$ and $|W_s-X_s|\le1$ otherwise. An inserted record contributes at most $\gamma/2-\gamma^2/8\le\gamma$ to $\log Z_t$ and nothing to $\log L_t$. Summing, $\log Z_t\le\log L_t+\gamma\Delta_t+\gamma N_t$, and we conclude as in (a).

\emph{Statement of (c).} Let $\tilde\mu_u=\E[X_u\mid X_1,\dots,X_{u-1}]$, $\tilde\Delta_u=\sum_{u'\le u}(\tilde\mu_{u'}-\frac12)$ and $h(t)$ be the number of non-inserted records among the first $t$. Part (c) states that, against every attacker in $\classO$ and with the cost $\gamma$, $\P\big(\exists\,t:\ \tilde\Delta_{h(t)}\le0\ \text{and}\ \Bhat^Z_t\ge N_t\big)\le\alpha$.

\emph{(c) Class $\classO$, cost $\gamma$.} Let $\tilde L_u=\exp\big(\gamma\sum_{u'\le u}(X_{u'}-\tilde\mu_{u'})-\gamma^2u/8\big)$, a nonnegative supermartingale for the natural filtration of $(X_u)$ by \cref{lem:hoeffding}, whose law does not depend on the attacker. The term of an uncorrupted value in $\log Z_t$ depends on that value alone, so its position in the observed sequence does not matter. As in (b), $\log Z_t\le\log\tilde L_{h(t)}+\gamma\tilde\Delta_{h(t)}+\gamma N_t$ pathwise. Hence $\{\exists t:\tilde\Delta_{h(t)}\le0,\ \Bhat^Z_t\ge N_t\}\subseteq\{\exists u:\sum_kw_k\tilde L^{(k)}_u\ge1/\alpha\}$, an event of the uncorrupted sequence alone, of probability at most $\alpha$ by \cref{lem:ville}. As in \cref{thm:validity}(c), the constants $\gamma_k$ must not depend on the observed records.

\emph{(d) Validity in the sense of \eqref{eq:validity}.} Under $H_0^{ij}$, $\mu_s\le\frac12$ on $\{I_s=0\}$ for every $s$, so $\Delta_t\le0$ for every $t$, and the event of \eqref{eq:validity} is that of \eqref{eq:validity-avg}. Under the null of $\classO$, $\tilde\Delta_u\le0$ for every $u$ likewise.

\emph{The rate.} By the strong law of large numbers, $t^{-1}\log Z^{(k)}_t\to\gamma_k\delta-\gamma_k^2/8$ almost surely, since $\E[W-\frac12]=\delta$ whether or not there are ties. The sandwich in the proof of \cref{thm:upper} applies with $g(\lb_k)$ replaced by $g^Z(\gamma_k)=\gamma_k\delta-\gamma_k^2/8$ and $c(\lb_k)$ by $c^Z(\gamma_k)$, which gives the limit $\rho^Z$. The two bounds follow at $k=m$ from $g^Z(\gamma)/(\gamma/2-\gamma^2/8)\ge g^Z(\gamma)/(\gamma/2)=2\delta-\gamma/4$ when $g^Z(\gamma)\ge0$, from $g^Z(\gamma)/\gamma=\delta-\gamma/8$, and from the fact that both right-hand sides are negative when $g^Z(\gamma_m)<0$. The limits $\rho^Z\le2\delta$ and $\rho^Z\le\delta$ follow from the coupling bound of \eqref{eq:lower-coupling}, which holds with ties (\cref{app:lower}). That bound is proved with an insertion-only attacker in the case of $\classP$ and with an attacker in $\classV$ in the case of $\classV$, and $\Bhat^Z_t$ is valid against both by (d).

\emph{Sharpness of the insertion cost and a lower bound for the replacement cost.} $B$ forged wins inserted before the first uncorrupted record give $\log Z_B=B(\gamma/2-\gamma^2/8)$ and $N_B=B$. So with any cost $c<\gamma/2-\gamma^2/8$ the lead tolerance makes a false statement with probability one once $B\ge\log(1/\alpha)/(\gamma/2-\gamma^2/8-c)$, as in \cref{thm:sharp}(a). Replacing the first $B$ uncorrupted losses by wins raises $\log Z$ by exactly $\gamma$ per flip, and the argument of \cref{thm:sharp}(b) shows that every cost below $\gamma-\gamma^2/4=2\cins^Z(\gamma)$ then fails with probability tending to one. The cost $\gamma$ is the largest increase of $\log Z_t$ that one flip can cause, and it is valid by (b). This argument does not show that every cost between $\gamma-\gamma^2/4$ and $\gamma$ fails.

\emph{Why (a) excludes predictable replacement.} Under $H_0^{ij}$ a record overwritten unseen is a null record anyway, which is why \cref{thm:validity}(a) covers it. For the lead claim its hidden uncorrupted value counts in $\Delta_t$, and an attacker that overwrites the records it expects to be weak can hold $\Delta_t$ at zero while the wealth grows. The lead certificate therefore charges $\gamma/2-\gamma^2/8$ for a forged record and $\gamma$ for an altered one. Fix $0<\xi<(4-\gamma_1)/8$ and let the uncorrupted values be independent with $\mu_s=1-\xi$ at odd $s$ and $\mu_s=\xi$ at even $s$, so that $\Delta_t=0$ at every even $t$. A predictable attacker overwrites every even record by a win without looking at it. Then $N_t=t/2$ at even $t$, and by the strong law $t^{-1}\sum_{s\le t}(W_s-\frac12)\to\frac12-\frac\xi2$ almost surely, so $t^{-1}\big(\log Z^{(\gamma)}_t-(\gamma/2-\gamma^2/8)N_t\big)\to\frac{\gamma}{16}(4-8\xi-\gamma)$, which is positive for every $\gamma\le\gamma_1$. Every bet of the grid thus has $Z^{(k)}_te^{-(\gamma_k/2-\gamma_k^2/8)N_t}\to\infty$, and $\Bhat^Z_t\ge N_t$ at all large even $t$ almost surely while $\Delta_t=0$. The margin over the uncorrupted records is $(\frac12-\xi)\lceil t/2\rceil>0$, so the variant at the end of (a) is not contradicted.

\begin{corollary}[Simultaneous lead certificates]\label{cor:avg-simul}
Let the pair of each record be fixed before its value is produced and let $\mathcal M$ be a set of $M$ ordered pairs fixed in advance. Let $\Delta^{ij}_t$ be the uncorrupted margin of \eqref{eq:delta} of pair $(i,j)$, let $\Bhat^{Z,ij}_t$ be its lead tolerance computed from its records with $M/\alpha$ in place of $1/\alpha$, and let $N_t$ count the corrupted records on all pairs. Against the attackers and costs of \cref{thm:average}(a,b), $\P\big(\exists\,t,\ \exists(i,j)\in\mathcal M:\ \Delta^{ij}_t\le0\ \text{and}\ \Bhat^{Z,ij}_t\ge N_t\big)\le\alpha$. The same holds for $V$ private variants of one model with $V/\alpha$.
\end{corollary}

The set $S^\ast$ of true nulls in \cref{prop:closed} is fixed, while the pairs with $\Delta^{ij}_t\le0$ change with $t$, so closed testing does not transfer. A Bonferroni correction does, at a price of $\log M$ nats per pair, and in E3c closed testing gained at most one record over it.

\emph{Proof of \cref{cor:avg-simul}.} Let $t$ run over the records of all pairs with the global filtration $(\cG_t)$, as in the proof of \cref{prop:closed}. For a pair $(i,j)$ define $L^{(ij,k)}$ as in (a) or (b), with factor $1$ on the records of other pairs. Since the pair of record $t$ is $\cG_{t-1}$-measurable, the computations of (a) and (b) show that $L^{(ij,k)}$ is a nonnegative supermartingale and that, on $\{\Delta^{ij}_t\le0\}$, $Z^{(ij,k)}_t\le e^{c_kN^{ij}_t}L^{(ij,k)}_t\le e^{c_kN_t}L^{(ij,k)}_t$, where $N^{ij}_t\le N_t$ counts the corrupted records of the pair. By the reduction above with $M/\alpha$ in place of $1/\alpha$, $\{\Delta^{ij}_t\le0,\ \Bhat^{Z,ij}_t\ge N_t\}\subseteq\{\sum_kw_kL^{(ij,k)}_t\ge M/\alpha\}$ for every $t$, and \cref{lem:ville} bounds the probability of the union over $t$ by $\alpha/M$. A union bound over the $M$ pairs gives $\alpha$. For $V$ private variants of one model against $j$, the variants play the role of the pairs. \hfill$\square$

\subsection{A deterministic tolerance for a completed set of records}\label{app:finite}

The certificates above are statements about the normal process that produced the records. A completed set of records, such as a released benchmark or the votes of a closed period, also supports a statement about itself, with no probability and no sampling assumption.

\begin{proposition}[Deterministic tolerance of a completed set]\label{prop:finite}
Let $n$ records be observed with values $W_1,\dots,W_n\in[0,1]$ and let $\Delta=\sum_{s\le n}(2W_s-1)>0$, not the margin of \eqref{eq:delta}. If at most $B\le\lceil\Delta\rceil-1$ of the records were inserted and none replaced, the mean of the uncorrupted values among them exceeds $\frac12$. If at most $B\le\lceil\Delta/2\rceil-1$ were replaced and none inserted, the mean of the uncorrupted values of all $n$ records exceeds $\frac12$. For records in $\{0,\frac12,1\}$, $\Delta$ is the number of wins minus losses and neither bound can be raised. If $\Delta$ of the observed wins were inserted, the uncorrupted records have mean $\frac12$ when any remain, and none remain when $\Delta=n$. If $\lceil\Delta/2\rceil$ of them replaced losses, the mean of the uncorrupted values is at most $\frac12$.
\end{proposition}

\begin{proof}
An inserted record contributes at most $1$ to $\Delta$, so the sum of $2W_s-1$ over the uncorrupted records is at least $\Delta-B$, which is positive when $B<\Delta$, that is, when $B\le\lceil\Delta\rceil-1$. A replacement changes $2W_s-1$ by at most $2$, so the same sum taken over the uncorrupted values of all $n$ records is at least $\Delta-2B$, which is positive when $B<\Delta/2$, that is, when $B\le\lceil\Delta/2\rceil-1$. For the last claim, the observed set has $\Delta$ more wins than losses, so at least $\Delta$ wins. If $\Delta$ of them were inserted, the uncorrupted records have as many wins as losses, so their mean is $\frac12$ when any remain (none remain when $\Delta=n$). If $\lceil\Delta/2\rceil$ of them replaced uncorrupted losses, the uncorrupted values have $\Delta-2\lceil\Delta/2\rceil\le0$ more wins than losses, and their mean is at most $\frac12$.
\end{proof}

For records in $\{0,\frac12,1\}$ a completed set thus withstands exactly $\Delta-1$ forged records and $\lceil\Delta/2\rceil-1$ replaced ones. If a contaminated item can change a record by at most $\frac12$, as when only one model's answer is altered, the guaranteed count for replaced items rises to $\lceil\Delta\rceil-1$. For general values in $[0,1]$ a given observed set need not attain the bounds. The certificates of \eqref{eq:bhat-grid} and \eqref{eq:zwealth} speak about the normal process. For records in $\{0,\frac12,1\}$ they are at most $\Delta-1$ and $\lceil\Delta/2\rceil-1$ as well. Indeed $\log K^{(k)}_n\le\Delta\,\cins(\lb_k)$ because $\log(1-\lb_k/2)\le-\log(1+\lb_k/2)$, and $\log Z^{(\gamma)}_n\le\Delta\,\cins^Z(\gamma)$ because $n\ge\Delta$. Each replacement cost is at least twice the insertion cost. So the deflated average is at most $1<1/\alpha$ at $b=\Delta$ with the insertion cost and at $b=\lceil\Delta/2\rceil$ with the replacement cost. E3 and E7 report the gap. For general values in $[0,1]$ a certificate and the deterministic bound concern different statements and need not be ordered.

\subsection{Attacks on correct claims}\label{app:power}

An attacker may also work against a correct claim, by inserting losses for the leader or by replacing its wins with losses, as in E6. Let $d_{\rm ins}(\lb)=-\log(1-\lb/2)=\crep(\lb)-\cins(\lb)$, the most that one inserted record can lower the log-wealth of a constant bet $\lb$.

\begin{proposition}[Price of attacking a correct claim]\label{prop:power}
Consider the tolerance of \eqref{eq:bhat-grid} with constant bets and an attacker in $\classV$ that corrupts $B$ of the first $t$ records. Let $L^{(k)}_t$ be the uncorrupted wealth of bet $k$, as in the proof of \cref{thm:validity}, and let $\Bhat^{(k)}_t(\alpha')=\lfloor\log(\alpha'L^{(k)}_t)/c(\lb_k)\rfloor$ be its single-bet tolerance at level $\alpha'$ on the uncorrupted values.
(a) Pathwise, $\log K^{(k)}_t\ge\log L^{(k)}_t-Bd_k$, where $d_k=d_{\rm ins}(\lb_k)$ if every corrupted record is inserted and $d_k=\crep(\lb_k)$ otherwise. Hence $\Bhat_t\ge\Bhat^{(k)}_t(\alpha w_k)-\lceil Bd_k/c(\lb_k)\rceil$ for every $k$ for which the right-hand side is nonnegative.
(b) With the cost $\crep$, replacing $B$ wins by losses lowers the tolerance by exactly $B$: $\Bhat_t=\Bhat^{\rm unc}_t-B$ if $\Bhat^{\rm unc}_t\ge B$ and $\Bhat_t=-1$ otherwise, where $\Bhat^{\rm unc}_t$ is the tolerance of the uncorrupted values.
(c) Suppose that the corrupted records lie among the first $t_0$ records, that the later uncorrupted values are independent of the first $t_0$ records and i.i.d.\ as in \cref{thm:upper} with $g(\lb_k)>0$, and that $\log L^{(k)}_{t_0}\ge\log\frac{1}{\alpha w_k}+b\,c(\lb_k)$ for an integer $b\ge0$. Then the number of records after $t_0$ until $\Bhat_t\ge b$ has expectation at most $\big(Bd_k+\cins(\lb_k)\big)/g(\lb_k)$.
\end{proposition}

\emph{(a)} The factor of a record lies in $[\mm,\Mp]$ with $\mm=1-\lb_k/2$ and $\Mp=1+\lb_k/2$. An inserted record contributes a factor of at least $\mm=e^{-d_{\rm ins}(\lb_k)}$ to $K^{(k)}_t$ and none to $L^{(k)}_t$. A replaced record changes its factor by a ratio of at least $\mm/\Mp=e^{-\crep(\lb_k)}$, and $d_{\rm ins}\le\crep$. Multiplying over the $B$ corrupted records gives the first claim. For the second, $\Phi_t(b)\ge w_kK^{(k)}_te^{-bc(\lb_k)}\ge w_kL^{(k)}_te^{-Bd_k-bc(\lb_k)}$, which is at least $1/\alpha$ at $b=\lfloor(\log(\alpha w_kL^{(k)}_t)-Bd_k)/c(\lb_k)\rfloor$, and $\lfloor x-y\rfloor\ge\lfloor x\rfloor-\lceil y\rceil$.

\emph{(b)} A win replaced by a loss changes the factor of every bet $k$ from $\Mp$ to $\mm$, a ratio of $e^{-\crep(\lb_k)}=e^{-c_k}$. So $K^{(k)}_t=L^{(k)}_te^{-Bc_k}$ for every $k$, hence $\Phi_t(b)=\Phi^{\rm unc}_t(b+B)$ with $\Phi^{\rm unc}$ the deflated average of the uncorrupted values, and the claim follows from \eqref{eq:bhat-grid}.

\emph{(c)} Write $z=\log\frac1{\alpha w_k}+b\,c(\lb_k)$. If $\log K^{(k)}_{t_0}\ge z$, then $\Phi_{t_0}(b)\ge w_kK^{(k)}_{t_0}e^{-bc(\lb_k)}\ge1/\alpha$, so $\Bhat_{t_0}\ge b$ and there is nothing to prove. Otherwise let $\sigma=\inf\{t>t_0:\log K^{(k)}_t\ge z\}$. At $\sigma$, $\Phi_\sigma(b)\ge w_kK^{(k)}_\sigma e^{-bc(\lb_k)}\ge1/\alpha$, so $\Bhat_\sigma\ge b$. After $t_0$, $\log K^{(k)}_t$ grows by independent increments with mean $g(\lb_k)>0$ that are bounded above by $\cins(\lb_k)$. By (a) and the assumption, $z-\log K^{(k)}_{t_0}\le Bd_k$. Wald's identity at the bounded stopping time $\sigma\wedge n$ gives $g(\lb_k)\,\E[(\sigma\wedge n)-t_0]\le Bd_k+\cins(\lb_k)$, as in the proof of \eqref{eq:tau-upper}, and letting $n\to\infty$ completes the proof. \hfill$\square$

The ratios in (a) explain E6. Per record, an inserted loss removes between $1$ and $\Mp/\mm$ units of bet $k$'s tolerance with the insertion cost, since $\cins\le d_{\rm ins}\le\cins\Mp/\mm$, and between $\frac12$ and $1$ unit with the replacement cost, since $d_{\rm ins}=\crep-\cins\ge\crep/2$. A replaced win removes $\crep/\cins\ge2$ units with the insertion cost and exactly one unit with the replacement cost, by (b). The small bets that carry the tolerance of a well-separated pair are near the lower ends of these ranges.

\section{Experimental Details and Additional Results}\label{app:experiments}

\subsection{Overview of the experiments}\label{app:overview}

\Cref{tab:map} lists the experiments with the question each one answers, the results it checks, its main takeaway, and the tables and figures that report it.

\begin{table}[htbp]
\caption{Map of the experiments. Checks: the theorems and propositions that each experiment tests or illustrates. Results: where its numbers are reported.}
\label{tab:map}
\centering\small\renewcommand{\arraystretch}{1.2}
\begin{tabularx}{\textwidth}{@{}l>{\raggedright\arraybackslash}X>{\raggedright\arraybackslash}p{2.9cm}>{\raggedright\arraybackslash}X>{\raggedright\arraybackslash}p{2.5cm}@{}}
\toprule
Exp. & Question & Checks & Main takeaway & Results \\
\midrule
E1 & Are the certificates valid against bursts, spread forgeries and flips? Is each cost needed? & \cref{thm:validity,thm:sharp}, \cref{prop:perstep} & The replacement certificate is valid under all four attacks; the insertion certificate fails only under flips; the baselines fail under bursts. & \cref{fig:validity}, \cref{tab:e1} \\
E2 & Does the tolerance grow at the optimal rates? & \cref{thm:upper,thm:lower}, \cref{cor:explicit,cor:lower-mean} & Tolerance per record close to $2\delta$ and $\delta$; time to certify close to the upper bound. & \cref{fig:rates}(a,b), \cref{tab:tau-union} \\
E3a & How large are the tolerances on real votes? & \cref{thm:validity,thm:average}, \cref{prop:finite} & Separated pairs withstand about $2{,}000$ forged votes; close pairs are not certified. & \cref{fig:rates}(c), \cref{tab:arena-pairs} \\
E3b & Do rigged real votes fool the baselines? & \cref{thm:validity,thm:sharp} & $400$ forged votes in a burst: the Wald interval and undeflated betting err in $98$ to $100\%$ of replays, our certificates in at most $2.2\%$ at any budget. & \cref{tab:rigging,tab:rigging-full} \\
E3c & What holds when the whole leaderboard is claimed at once? & \cref{prop:closed} & $613$ of $16{,}512$ ordered pairs certified; GPT-4o is not certified above Claude-3.5-Sonnet. & \cref{fig:matrix} \\
E4 & Is the price of selecting among $V$ variants needed and enough? & \cref{prop:selection} & Best of $27$ null variants: false statement $0.52$ without correction, $0.025$ with $V/\alpha$. & \cref{tab:selection} \\
E5 & What happens when the better model changes? & \cref{thm:average}, identity of \eqref{eq:netwins} & Both certificates return to $-1$ about $420$ records after the switch. & \cref{fig:regime} \\
E6 & What can an attacker do to a correct claim? & \cref{prop:power} & A flipped win removes one flipped-vote unit; de-certifying costs about the tolerance. & \cref{tab:suppress} \\
E7 & Do the results carry over to a static benchmark? & \cref{thm:validity,thm:sharp,thm:average}, \cref{prop:finite} & Tolerances are $91$ to $94\%$ of the benchmark's own count; only the replacement certificate survives flips. & \cref{tab:mmlu}, \cref{fig:benchmark}(a,b) \\
E8 & Are frontier-model rankings on benchmarks certified? & \cref{prop:closed}, \cref{cor:items} & On all five benchmarks no top-two ordering is certified simultaneously. & \cref{tab:frontier}, \cref{fig:benchmark}(c), \cref{fig:frontier} \\
E9 & What does the lead claim cost? & \cref{thm:upper,thm:average} & Equal to the edge tolerance on binary records; $4$ to $11\%$ lower with $30\%$ ties. & \cref{tab:average} \\
E10 & How far is a fixed-time sensitivity count? & \cref{thm:validity}, \cref{prop:finite} & The edge tolerance is $90$ to $97\%$ of the Wald count, which attacks can time. & \cref{tab:arena-pairs,tab:mmlu} \\
\bottomrule
\end{tabularx}
\end{table}

\textbf{Grid.} All experiments use the grid of \cref{sec:method}, $\lb_k=2^{-(k-1)/2}$ for $k=1,\dots,14$ with equal weights. Its $\sqrt2$ spacing puts a bet in $[3\delta/\sqrt2,3\delta]$ for every margin $\delta$ from about $0.004$ to $\frac13$, as \cref{cor:explicit} requires. All experiments use $\alpha=0.05$ and fixed seeds, and \cref{fig:validity} also shows our single-bet certificate of \eqref{eq:bhat-single} with $\lambda=0.2$.

\textbf{Baselines.} \emph{Undeflated betting} compares the wealth, not deflated, with $1/\alpha$ \citep{waudbysmith2024}. It uses the single bet $\lambda=0.2$ in simulations and the average over the grid in replays. \emph{Per-step robust}, a simplified form of the test of \citet{saha2024}, tests the null mean $\frac12+\varepsilon$ \citep[Thm.~6]{wang2023huber}. This is valid when each record is corrupted with conditional probability at most $\varepsilon$. It is favorably calibrated with $\varepsilon=B/n$ from the attacker's true budget $B$ and the number $n$ of uncorrupted values. \emph{Wald interval} is the $95\%$ normal interval for the mean outcome, recomputed every $50$ records.
\subsection{Reproducibility}\label{app:repro}
All code, the benchmark data we use, the saved results of every experiment and the scripts that rebuild every figure and table are at \url{https://github.com/HamedKhosravi99/corruption-tolerance}. Each experiment runs from one or two scripts with fixed seeds. The repository's README maps every figure, table and number of the paper to the script and the saved files that produce it. Rebuilding all figures and tables from the saved results takes about a minute. The per-item results of HELM and SWE-bench Verified are public and included. The Chatbot Arena battle file is public but states no license, so the repository gives the commands that download and extract it instead of a copy. The code needs Python with NumPy, SciPy, pandas and Matplotlib, with versions pinned in the repository. The experiments ran on a 12-core workstation, a laptop and four-core jobs on a Linux cluster, and the longest took about an hour.

\subsection{E1: validity}

Uncorrupted values are independent $\Bern(\frac12)$, with $n=5{,}000$ values per repetition and $10{,}000$ repetitions per cell. There are four attackers, and \cref{fig:validity} shows three of them. \emph{Burst at start} inserts $B$ forged wins before the first uncorrupted record. \emph{Burst at trigger} inserts $B$ forged wins as soon as the undeflated log-wealth with $\lambda=0.2$ first reaches $\log(1/\alpha)-B\cins(0.2)$. This is the earliest moment at which the burst carries the undeflated wealth to $1/\alpha$. For the single-bet insertion certificate the timing of forged wins does not matter. \emph{Spread} inserts a forged win before each uncorrupted record with probability $B/n$ until the budget is spent. \emph{Flip losses} replaces the first $B$ uncorrupted losses by wins. For constant bets and binary records, the last attacker is the optimal offline replacement attacker (\cref{thm:sharp}(d)). Undeflated betting and the per-step robust test use the single bet $\lambda=0.2$. \Cref{tab:e1} lists the rates for $B\in\{0,10,20,40,100\}$, and the results file in the repository also has $B=5$, $30$ and $60$. \Cref{fig:validity} plots burst at trigger, flip losses and spread at all eight budgets, and burst at start appears only in \cref{tab:e1}.

\begin{table}[htbp]
\caption{E1: probability that a method makes a false statement at some time on null data, for each attack and attacker budget $B$ ($5{,}000$ uncorrupted $\Bern(\frac12)$ records; $10{,}000$ repetitions; standard error at most $0.005$). A method is valid if its entry stays at or below $\alpha=0.05$. Undefl.\ and Per-step: undeflated betting and the per-step robust test, with the single bet $\lambda=0.2$. Ins.\ and Rep.: our insertion and replacement certificates, with this bet (1 bet) or with the grid of \cref{sec:method} (grid). Rows: the four attacks described above.}
\label{tab:e1}
\centering\small
\begin{tabular}{lrrrrrrr}
\toprule
attack & $B$ & Undefl. & Per-step & Ins.\ (1 bet) & Ins.\ (grid) & Rep.\ (1 bet) & Rep.\ (grid) \\
\midrule
burst at start & 0 & 0.049 & 0.049 & 0.049 & 0.028 & 0.049 & 0.028 \\
burst at start & 10 & 0.127 & 0.110 & 0.052 & 0.029 & 0.019 & 0.008 \\
burst at start & 20 & 0.317 & 0.263 & 0.048 & 0.029 & 0.006 & 0.003 \\
burst at start & 40 & 1.000 & 1.000 & 0.046 & 0.028 & 0.001 & 0.001 \\
burst at start & 100 & 1.000 & 1.000 & 0.051 & 0.032 & 0.000 & 0.000 \\
\midrule
burst at trigger & 0 & 0.048 & 0.048 & 0.048 & 0.028 & 0.048 & 0.028 \\
burst at trigger & 10 & 0.124 & 0.104 & 0.047 & 0.027 & 0.015 & 0.014 \\
burst at trigger & 20 & 0.326 & 0.266 & 0.046 & 0.027 & 0.005 & 0.004 \\
burst at trigger & 40 & 1.000 & 1.000 & 0.046 & 0.028 & 0.001 & 0.000 \\
burst at trigger & 100 & 1.000 & 1.000 & 0.044 & 0.029 & 0.000 & 0.000 \\
\midrule
spread & 0 & 0.052 & 0.052 & 0.052 & 0.033 & 0.052 & 0.033 \\
spread & 10 & 0.057 & 0.044 & 0.048 & 0.028 & 0.043 & 0.024 \\
spread & 20 & 0.062 & 0.039 & 0.049 & 0.030 & 0.040 & 0.023 \\
spread & 40 & 0.073 & 0.028 & 0.045 & 0.027 & 0.029 & 0.017 \\
spread & 100 & 0.156 & 0.014 & 0.049 & 0.029 & 0.015 & 0.009 \\
\midrule
flip losses & 0 & 0.049 & 0.049 & 0.049 & 0.027 & 0.049 & 0.027 \\
flip losses & 10 & 0.354 & 0.328 & 0.137 & 0.449 & 0.046 & 0.028 \\
flip losses & 20 & 0.982 & 0.979 & 0.369 & 0.935 & 0.047 & 0.028 \\
flip losses & 40 & 1.000 & 1.000 & 0.947 & 1.000 & 0.050 & 0.030 \\
flip losses & 100 & 1.000 & 1.000 & 1.000 & 1.000 & 0.048 & 0.028 \\

\bottomrule
\end{tabular}
\end{table}

\subsection{E2: growth of the tolerance and time to certify}

Uncorrupted records are independent $\Bern(\frac12+\delta)$. The horizon is $60{,}000$ records for $\delta=0.02$, $20{,}000$ for $\delta=0.05$ and $8{,}000$ for $\delta=0.1$. Each $\delta$ has $2{,}000$ repetitions. Every repetition reached every tolerance $b$ in \cref{tab:tau-union} within the horizon. The column ``union'' is the mean of $\tau_b$ for the tolerance $\max\{-1,\max_k\lfloor\log(\alpha w_kK^{(k)}_t)/c_k\rfloor\}$, which takes the best bet at level $\alpha w_k$ instead of averaging the deflated wealths. Averaging the deflated wealths needs $2.5$--$22\%$ fewer records on average than the union variant. The upper bound of \eqref{eq:tau-upper} also applies to the union variant, since its $\tau_b$ is the minimum over $k$ of the $\sigma_k$ in the proof of \cref{thm:upper}. Five of its $48$ computed sample means exceed the bound, each by less than $0.3\%$ and by less than one Monte Carlo standard error. Three of them are among the $30$ shown in \cref{tab:tau-union}, and one of these three equals the bound after rounding.

\begin{table}[htbp]
\caption{E2: number of records $\tau_b$ until the tolerance first reaches $b$, on uncorrupted i.i.d.\ records with mean outcome $\frac12+\delta$ and no attack ($2{,}000$ repetitions). Mean $\tau_b$ and median: our tolerance of \eqref{eq:bhat-grid}. Union: mean $\tau_b$ of the union-over-bets variant defined above. Upper and lower: the bound of \eqref{eq:tau-upper} of \cref{thm:upper} and the bound of \cref{cor:lower-mean}. Cost: insertion (ins.) or replacement (rep.). Means close to the upper bound show that it is nearly attained.}
\label{tab:tau-union}
\centering\small
\begin{tabular}{llrrrrrr}
\toprule
$\delta$ & cost & $b$ & mean $\tau_b$ & union & upper \eqref{eq:tau-upper} & lower (Cor.~\ref{cor:lower-mean}) & median \\
\midrule
0.02 & ins. & 0 & 4,866 & 6,268 & 7,175 & 3,744 & 4,146 \\
0.02 & ins. & 10 & 5,603 & 6,905 & 7,721 & 3,744 & 4,968 \\
0.02 & ins. & 50 & 7,692 & 9,011 & 9,456 & 3,744 & 7,156 \\
0.02 & ins. & 100 & 9,908 & 11,296 & 11,475 & 3,744 & 9,373 \\
0.02 & ins. & 400 & 20,936 & 22,550 & 22,505 & 8,411 & 20,429 \\
0.02 & rep. & 0 & 4,866 & 6,268 & 7,175 & 3,744 & 4,146 \\
0.02 & rep. & 10 & 6,175 & 7,512 & 8,257 & 3,744 & 5,502 \\
0.02 & rep. & 50 & 9,990 & 11,372 & 11,539 & 3,744 & 9,462 \\
0.02 & rep. & 100 & 13,930 & 15,458 & 15,642 & 3,835 & 13,318 \\
0.02 & rep. & 400 & 34,212 & 36,080 & 36,089 & 16,803 & 33,888 \\
0.05 & ins. & 0 & 807 & 1,021 & 1,158 & 598 & 705 \\
0.05 & ins. & 10 & 1,041 & 1,241 & 1,329 & 598 & 958 \\
0.05 & ins. & 50 & 1,715 & 1,940 & 2,015 & 598 & 1,637 \\
0.05 & ins. & 100 & 2,462 & 2,707 & 2,734 & 773 & 2,397 \\
0.05 & ins. & 400 & 6,327 & 6,651 & 6,670 & 3,376 & 6,266 \\
0.05 & rep. & 0 & 807 & 1,021 & 1,158 & 598 & 705 \\
0.05 & rep. & 10 & 1,249 & 1,458 & 1,516 & 598 & 1,168 \\
0.05 & rep. & 50 & 2,509 & 2,755 & 2,779 & 724 & 2,444 \\
0.05 & rep. & 100 & 3,879 & 4,164 & 4,216 & 1,538 & 3,809 \\
0.05 & rep. & 400 & 11,239 & 11,640 & 11,630 & 6,732 & 11,148 \\
0.10 & ins. & 0 & 208 & 258 & 292 & 149 & 177 \\
0.10 & ins. & 10 & 303 & 357 & 374 & 149 & 279 \\
0.10 & ins. & 50 & 608 & 669 & 675 & 184 & 591 \\
0.10 & ins. & 100 & 944 & 1,011 & 1,016 & 390 & 927 \\
0.10 & ins. & 400 & 2,739 & 2,836 & 2,836 & 1,698 & 2,729 \\
0.10 & rep. & 0 & 208 & 258 & 292 & 149 & 177 \\
0.10 & rep. & 10 & 400 & 455 & 472 & 149 & 382 \\
0.10 & rep. & 50 & 978 & 1,048 & 1,056 & 364 & 956 \\
0.10 & rep. & 100 & 1,609 & 1,693 & 1,699 & 773 & 1,596 \\
0.10 & rep. & 400 & 5,086 & 5,216 & 5,219 & 3,376 & 5,082 \\
\bottomrule
\end{tabular}
\end{table}

\subsection{E3: Chatbot Arena}

The three parts of E3 are the tolerances per pair (E3a), a rigging replay (E3b) and a simultaneous matrix (E3c).

\textbf{Arena data.} The file \texttt{clean\_battle\_20240814\_public.json} (2.1\,GB) is public on the maintainers' storage bucket \citep{chatbotarena_data} and is used by the official Arena leaderboard notebook. For each vote it contains the two model names, the winner, a Unix timestamp, a hashed voter identifier and metadata. The winner is coded as \texttt{model\_a}, \texttt{model\_b}, \texttt{tie} or \texttt{tie (bothbad)}. The file has no prompts or responses. We keep the model names, winner, timestamp and anonymity flag (all $1{,}799{,}991$ votes are anonymous battles). No license is stated. \citet{chiang2024arena} describe an active rule that chooses the next pair from past votes; the rule in force during these votes is not published. The certificates do not depend on it. They need only that the pair of a battle is fixed before its vote is cast, which holds under any sampling rule, since the user votes after reading both answers. The official notebook also down-samples the votes on the $0.1\%$ most repeated prompts, which leaves $1{,}670{,}250$ votes. We keep all $1{,}799{,}991$, of which $35.7\%$ are ties. Both tie labels are coded $W=\frac12$.

\Cref{tab:arena-pairs} (E3a) gives the tolerances at the end of the data for the twenty most frequently compared pairs. Well-separated pairs are certified against up to about $2{,}000$ forged votes, while the closest pairs are not certified. The last column is $\Delta=n(2\bar W-1)$, wins minus losses as in \cref{prop:finite}, where $n$ is the number of votes on the pair and $\bar W$ their mean outcome. By \cref{prop:finite} the completed vote set itself withstands $\Delta-1$ forged or $\lceil\Delta/2\rceil-1$ flipped votes; the certificates are below these counts by $10\%$ and $12\%$ for the first pair, the price of a claim about the normal process. The lead tolerances are $3$ to $9\%$ below the edge tolerances on the pairs certified above $500$ votes, the price of the $36\%$ ties under the exponential wealth. On marginal pairs the gap is larger, and the pair that the edge certificate certifies at $11$ forged votes is not certified by the lead certificate.

\Cref{tab:rigging-full} (E3b) gives the rigging replay of \cref{tab:rigging} for each of its three pairs. These are the three pairs that \cref{tab:arena-pairs} does not certify, and the trailing model is promoted. The promoted models have mean outcomes $0.488$--$0.495$ over $6{,}089$--$6{,}470$ votes. Each replay resamples the pair's own votes with replacement, so the claim ``the promoted model is better'' is false. Rigging that makes the baselines certify this false ordering leaves the matching certificate valid (the insertion certificate for forged votes and the replacement certificate for flipped ones). In the replay, a \emph{burst} inserts $B$ forged wins for the promoted model at a uniformly random position. \emph{Spread} inserts a forged win before each record with probability $B/n$ until $B$ are inserted. \emph{Flip} replaces the promoted model's first $B$ losses by wins. Undeflated betting there compares the grid average, not deflated, with $1/\alpha$, and the per-step test uses the single bet $\lambda=0.2$. The run covers $B$ up to $800$. \Cref{tab:rigging,tab:rigging-full} show $B=0$, $100$ and $400$.

\Cref{fig:matrix} (E3c) shows the tolerances of the edge claim ``row has ever had an edge over column'', computed from the votes of the last $120$ days and made simultaneous by closed testing over all sets of pairs (\cref{prop:closed}). The family is all $16{,}512$ ordered pairs of the $129$ models in the data, which does not depend on the votes. Of these pairs, $3{,}926$ have battles in the window and $613$ are certified. The ten displayed models are those with the highest net win rate (wins minus losses, per battle) among the $25$ most played models in the window. Because the guarantee covers every pair, the choice of which pairs to display may depend on the data. Closed testing certified at most one more record than a Bonferroni split of $\alpha$ over the $16{,}512$ pairs. Simultaneity has a price. For GPT-4o over Claude-3.5-Sonnet ($5{,}868$ battles), the single-pair tolerance at the end of the data is $53$ forged or $26$ flipped votes, while the simultaneous tolerance is $-1$. For GPT-4o over GPT-4-Turbo the single-pair tolerance is $345$ forged votes and the simultaneous one $225$. Fixing the family in advance also has a price. With a family of only the $90$ ordered pairs among the displayed ten, the tolerance for GPT-4o over GPT-4-Turbo would be $280$ and three more cells would be certified. That family is chosen from the same votes, so \cref{prop:closed} does not cover it. The difference is about $\log(16{,}512/90)\approx5.2$ nats of evidence, which the closed test spends on the additional pairs instead of on tolerance.

\begin{table*}[htbp]
\caption{E3a: tolerances on the real, time-ordered Chatbot Arena votes, for the twenty most frequently compared pairs. Each row concerns the edge and lead claims for the leader against the other model. An entry is the number of corrupted votes that a claim withstands after all $n$ votes of the pair, and $-1$ means not certified. Edge and lead: the edge tolerance of \eqref{eq:bhat-grid} and the lead tolerance of \cref{thm:average}, in forged and in flipped votes. Wald: the votes that the fixed-time $95\%$ Wald interval absorbs (E10). $\bar W$: mean outcome, with ties counted as $\frac12$. Ties: share of tied votes. $\Delta$: wins minus losses. Models are named by their Arena identifiers.}
\label{tab:arena-pairs}
\centering\scriptsize\setlength{\tabcolsep}{3pt}
\begin{tabular}{llrrrrrrrrrr}
\toprule
 & & & & & \multicolumn{2}{c}{edge} & \multicolumn{2}{c}{lead} & \multicolumn{2}{c}{Wald} & \\
\cmidrule(lr){6-7}\cmidrule(lr){8-9}\cmidrule(lr){10-11}
leader & other & votes & $\bar W$ & ties & forged & flipped & forged & flipped & forged & flipped & $\Delta$ \\
\midrule
llama-3-70b-instruct & llama-3-8b-instruct & 14,858 & 0.577 & 0.48 & 2,056 & 1,009 & 1,940 & 957 & 2,139 & 1,056 & 2,287 \\
claude-3-opus-20240229 & llama-3-70b-instruct & 11,971 & 0.554 & 0.35 & 1,053 & 517 & 984 & 485 & 1,143 & 565 & 1,304 \\
claude-3-opus-20240229 & claude-3-sonnet-20240229 & 9,513 & 0.579 & 0.38 & 1,288 & 630 & 1,217 & 598 & 1,363 & 672 & 1,495 \\
gpt-4-turbo-2024-04-09 & llama-3-70b-instruct & 8,940 & 0.571 & 0.35 & 1,065 & 521 & 1,003 & 493 & 1,141 & 562 & 1,275 \\
llama-3-70b-instruct & reka-core-20240501 & 8,446 & 0.518 & 0.39 & 80 & 39 & 18 & 9 & 158 & 78 & 298 \\
gemini-1.5-pro-api-0409-preview & llama-3-70b-instruct & 7,633 & 0.556 & 0.35 & 658 & 321 & 602 & 296 & 730 & 359 & 858 \\
claude-3-opus-20240229 & claude-3-haiku-20240307 & 7,510 & 0.599 & 0.37 & 1,301 & 635 & 1,238 & 608 & 1,366 & 672 & 1,480 \\
gemini-1.5-pro-api-0409-preview & llama-3-8b-instruct & 7,292 & 0.623 & 0.34 & 1,621 & 791 & 1,563 & 767 & 1,683 & 828 & 1,794 \\
claude-3-opus-20240229 & gpt-4-0125-preview & 7,190 & 0.515 & 0.33 & 11 & 5 & $-1$ & $-1$ & 87 & 43 & 222 \\
claude-3-opus-20240229 & command-r-plus & 6,753 & 0.581 & 0.35 & 919 & 448 & 864 & 423 & 983 & 484 & 1,098 \\
claude-3-sonnet-20240229 & claude-3-haiku-20240307 & 6,751 & 0.530 & 0.42 & 214 & 104 & 152 & 74 & 281 & 138 & 400 \\
gpt-4-0314 & gpt-4-0613 & 6,607 & 0.536 & 0.42 & 299 & 145 & 237 & 116 & 364 & 179 & 480 \\
gpt-4-turbo-2024-04-09 & claude-3-opus-20240229 & 6,470 & 0.505 & 0.35 & $-1$ & $-1$ & $-1$ & $-1$ & $-1$ & $-1$ & 69 \\
gpt-4o-2024-05-13 & llama-3-70b-instruct & 6,470 & 0.615 & 0.32 & 1,314 & 640 & 1,263 & 619 & 1,375 & 676 & 1,483 \\
llama-3-70b-instruct & snowflake-arctic-instruct & 6,418 & 0.650 & 0.34 & 1,771 & 862 & 1,714 & 839 & 1,826 & 897 & 1,923 \\
gpt-4-0125-preview & gpt-4-0613 & 6,395 & 0.637 & 0.35 & 1,595 & 776 & 1,536 & 752 & 1,651 & 811 & 1,749 \\
gemini-1.5-pro-api-0409-preview & claude-3-opus-20240229 & 6,324 & 0.512 & 0.34 & $-1$ & $-1$ & $-1$ & $-1$ & 19 & 9 & 146 \\
llama-3-8b-instruct & snowflake-arctic-instruct & 6,206 & 0.569 & 0.39 & 685 & 334 & 627 & 307 & 747 & 367 & 855 \\
claude-3-opus-20240229 & llama-3-8b-instruct & 6,175 & 0.615 & 0.32 & 1,259 & 613 & 1,209 & 592 & 1,319 & 648 & 1,424 \\
llama-3-70b-instruct & claude-3-sonnet-20240229 & 6,089 & 0.509 & 0.37 & $-1$ & $-1$ & $-1$ & $-1$ & $-1$ & $-1$ & 106 \\
\bottomrule
\end{tabular}
\end{table*}

\begin{table*}[htbp]
\caption{E3b: the rigging replay of \cref{tab:rigging} for each of the three pairs, as the probability of a false statement at $B=0$ / $100$ / $400$ votes ($1{,}000$ replays per cell). Promoted: the slightly worse model that the attacker favors. Burst: forged votes at one random time. Spread: forged votes at random times. Flip: the promoted model's losses turned into wins. Wald: the Wald interval recomputed every $50$ votes. Undefl.: undeflated betting with the grid average. Per-step: the per-step robust test. Ins.\ and Rep.: our insertion and replacement certificates. An entry at or below $0.05$ means valid.}
\label{tab:rigging-full}
\centering\scriptsize\setlength{\tabcolsep}{4pt}
\begin{tabular}{llccccc}
\toprule
promoted vs opponent & attack & Wald & Undefl. & Per-step & Ins. & Rep. \\
\midrule
Claude-3-Opus vs GPT-4-Turbo & burst & .09 / .51 / 1.0 & .00 / .28 / 1.0 & .02 / .17 / .30 & .00 / .02 / .01 & .00 / .02 / .01 \\
 & spread & .08 / .27 / 1.0 & .00 / .04 / .99 & .01 / .01 / .00 & .00 / .01 / .01 & .00 / .00 / .00 \\
 & flip & .09 / 1.0 / 1.0 & .01 / 1.0 / 1.0 & .01 / 1.0 / 1.0 & .01 / 1.0 / 1.0 & .01 / .01 / .01 \\
Claude-3-Opus vs Gemini-1.5-Pro-0409 & burst & .05 / .30 / 1.0 & .00 / .16 / .98 & .01 / .10 / .24 & .00 / .01 / .01 & .00 / .01 / .01 \\
 & spread & .05 / .11 / .97 & .01 / .01 / .84 & .00 / .00 / .00 & .01 / .00 / .01 & .01 / .00 / .00 \\
 & flip & .04 / 1.0 / 1.0 & .00 / 1.0 / 1.0 & .00 / 1.0 / 1.0 & .00 / 1.0 / 1.0 & .00 / .01 / .00 \\
Claude-3-Sonnet vs Llama-3-70B & burst & .06 / .40 / 1.0 & .00 / .22 / 1.0 & .01 / .12 / .30 & .00 / .01 / .01 & .00 / .01 / .01 \\
 & spread & .05 / .17 / 1.0 & .01 / .03 / .95 & .01 / .00 / .00 & .01 / .01 / .01 & .01 / .00 / .00 \\
 & flip & .07 / 1.0 / 1.0 & .01 / 1.0 / 1.0 & .01 / 1.0 / 1.0 & .01 / 1.0 / 1.0 & .01 / .00 / .01 \\
\bottomrule
\end{tabular}
\end{table*}

\begin{figure*}[htbp]
\centering
\includegraphics[width=\textwidth]{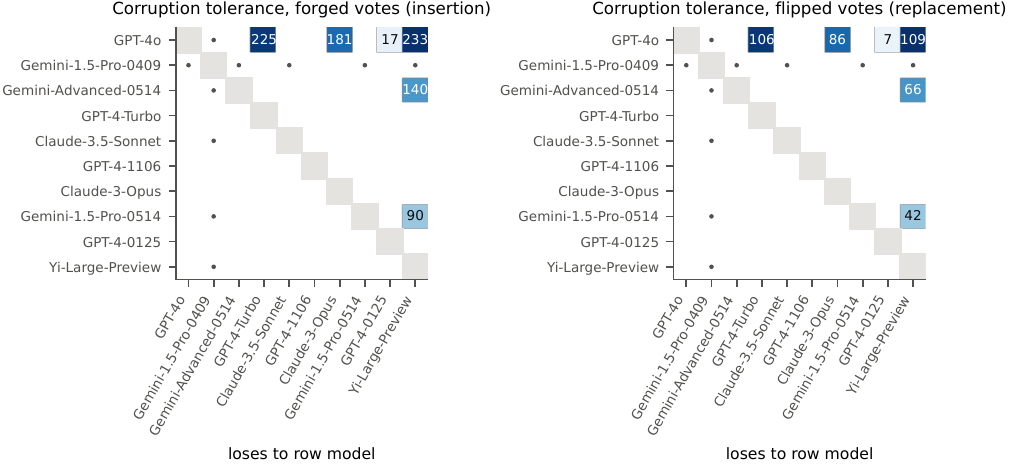}
\caption{E3c: corruption tolerance of the edge claim ``row model has ever had an edge over column model'' in the last $120$ days of the data, simultaneous over all $16{,}512$ ordered pairs of the $129$ models ($\alpha=0.05$). Shown are the ten models with the highest net win rate among the $25$ most played in the window. Empty cells: not certified. Dots: no battles. Gray diagonal: same model.}
\label{fig:matrix}
\end{figure*}

\subsection{E4: selection}

Each of $V$ variants has its own independent $\Bern(\frac12)$ records, with $n=5{,}000$ per variant and $4{,}000$ repetitions. The grid certificate is evaluated at $b=0$, and the corrected and uncorrected certificates use the same simulated records. The developer publishes the variant whose tolerance first becomes nonnegative. Publishing the best of $V$ variants that are no better than the opponent gives a false statement with probability $0.09$, $0.23$ and $0.52$ for $V=3$, $9$ and $27$ without correction. With the correction of \cref{prop:selection} these probabilities are $0.028$, $0.027$ and $0.025$.

\begin{table}[htbp]
\caption{E4: probability of a false statement when the developer publishes the best of $V$ variants, none of them better than the opponent ($4{,}000$ repetitions; standard errors in parentheses). Uncorrected: threshold $1/\alpha$. Corrected: threshold $V/\alpha$, as in \cref{prop:selection}. An entry at or below $\alpha=0.05$ means that the published certificate is valid.}
\label{tab:selection}
\centering\small
\begin{tabular}{rcc}
\toprule
$V$ & uncorrected & corrected \\
\midrule
1 & 0.029 (0.003) & 0.029 (0.003) \\
3 & 0.086 (0.004) & 0.028 (0.003) \\
9 & 0.233 (0.007) & 0.027 (0.003) \\
27 & 0.519 (0.008) & 0.025 (0.002) \\
\bottomrule
\end{tabular}
\end{table}

\subsection{E5: certificate persistence under a regime shift}

Uncorrupted records are binary. Model $i$ wins each of the first $1{,}000$ records with probability $0.6$ and each later record with probability $0.4$. There is no attack, and $2{,}000$ repetitions of $5{,}000$ records are run with the grid of \cref{sec:method}. The null hypothesis is false, so the claim is correct in the sense of \cref{sec:formulation}. The uncorrupted margin of \eqref{eq:delta} is positive before $t=2{,}000$ and at most zero from then on, so the lead claim is correct exactly before $t=2{,}000$. \Cref{fig:regime} shows the insertion and replacement tolerances and the empirical win rates. At the switch $99.85\%$ of the runs are certified, with median tolerances of $109$ forged and $52$ flipped records. After the switch the tolerance returns to $-1$ at a median of $420$ records (10--90\% range $240$ to $620$). One third of the runs are still certified $500$ records after the switch, and none $2{,}000$ records after it. The cumulative win rate crosses one half only at $t=2{,}000$. In $61\%$ of the runs the tolerance becomes nonnegative again for a few records after first reaching $-1$, when the average wealth $\Phi_t(0)$ crosses $1/\alpha$ once more. The renewed tolerance is then at most $14$ forged records. A second certificate starts at the prespecified time $t=1{,}001$ and sees only the later records. Its null hypothesis is true, and it certified nothing in any of the $2{,}000$ runs. The lead tolerance of \cref{thm:average}, computed on the same records with the same grid, is also shown. Its certified status agrees with that of the betting tolerance on all but $0.016\%$ of the (run, record) cells. Its median tolerances at the switch are $109$ forged and $51$ flipped records, and its return times are the same. In $0.05\%$ of the runs it is nonnegative at some $t\ge2{,}000$, where the lead claim is false, within the $5\%$ that \cref{thm:average} allows. The betting tolerance does so in the same runs. The experiment shows the meaning of a correct claim (\cref{app:testing}). It is not a validity failure.

\begin{figure*}[htbp]
\centering
\includegraphics[width=\textwidth]{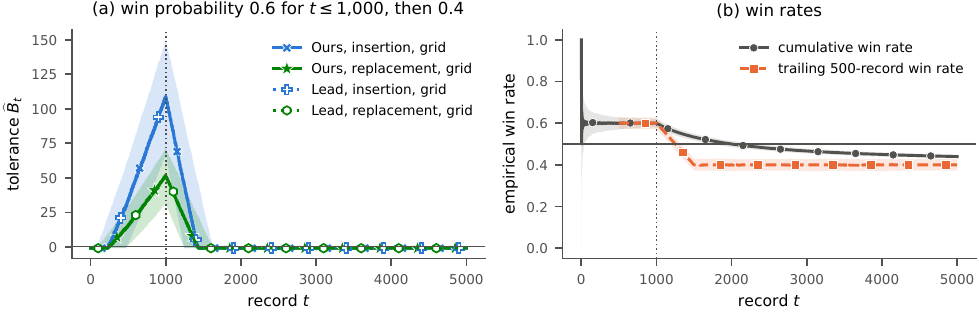}
\caption{E5: regime shift with no attack ($2{,}000$ repetitions, $\alpha=0.05$; median and $10$--$90\%$ band). (a) Betting tolerances of \eqref{eq:bhat-grid} with bands, and lead tolerances of \cref{thm:average} (dotted, medians only). (b) Cumulative and trailing empirical win rates. The vertical dotted line marks the switch.}
\label{fig:regime}
\end{figure*}

\subsection{E6: attacks on correct claims}

For the three most certified pairs of \cref{tab:arena-pairs} an attacker works against the leader on the real votes. It inserts $B$ losses for the leader in one burst, or it replaces $B$ of the leader's wins by losses. Each claim is certified against more than $1{,}600$ forged votes, so we treat it as correct. Then no statement can be false, and the question is how much tolerance the attacker destroys. By the identity of \eqref{eq:netwins} of \cref{app:testing}, the final tolerance does not depend on where the corrupted records are placed. \Cref{tab:suppress} reports the results. Each inserted loss removes about $1.05$ forged-vote units of tolerance and $0.51$ flipped-vote units. Each win replaced by a loss removes about $2.05$ forged-vote units and exactly $1.00$ flipped-vote unit. The attacker's budget that de-certifies a claim is therefore about the tolerance itself. For Llama-3-70B over Llama-3-8B it is $1{,}986$ inserted losses or $1{,}010$ replaced wins, against tolerances of $2{,}056$ and $1{,}009$. After a burst of $800$ inserted losses at the midpoint of the data, the tolerance of Llama-3-70B over Llama-3-8B never falls below zero. It is back at half of its pre-burst value after about $2{,}200$ uncorrupted votes. Llama-3-70B over Snowflake-Arctic is below zero for $191$ of the uncorrupted votes after the burst, and Gemini-1.5-Pro-0409 over Llama-3-8B for $108$. \Cref{prop:power} in \cref{app:proofs} accounts for these numbers. Per record, an inserted loss removes between $1$ and $(2+\lb_k)/(2-\lb_k)$ units of bet $k$'s forged-vote tolerance and between $\frac12$ and $1$ flipped-vote unit. A replaced win removes at least $2$ forged-vote units and exactly one flipped-vote unit. Recovery takes about $Bd/g$ uncorrupted records, with $d$ the loss per record and $g$ the growth rate of \eqref{eq:growth}.

\begin{table*}[htbp]
\caption{E6: attacks on correct claims, on the real Arena votes of the three most certified pairs. $n$: votes on the pair. $\Bhat_n$: tolerance after all $n$ votes, in forged and in flipped votes. De-certify by: the number of inserted losses, or of wins replaced by losses, that brings the forged-vote tolerance to $-1$. Loss of tolerance: forged-vote tolerance removed per inserted loss or per replaced win, at $B=100$. \Cref{prop:power} bounds these losses.}
\label{tab:suppress}
\centering\scriptsize\setlength{\tabcolsep}{3pt}
\begin{tabular}{lrrrrrrr}
\toprule
pair & $n$ & $\Bhat_n$ forged & $\Bhat_n$ flipped & de-certify by & de-certify by & loss of tolerance & loss of tolerance \\
 & & & & inserted losses & replaced wins & per inserted loss & per replaced win \\
\midrule
Llama-3-70B $>$ Llama-3-8B & 14,858 & 2,056 & 1,009 & 1,986 & 1,010 & 1.04 & 2.04 \\
Llama-3-70B $>$ Snowflake-Arctic & 6,418 & 1,771 & 862 & 1,689 & 863 & 1.06 & 2.06 \\
Gemini-1.5-Pro-0409 $>$ Llama-3-8B & 7,292 & 1,621 & 791 & 1,551 & 792 & 1.05 & 2.05 \\
\bottomrule
\end{tabular}
\end{table*}

\subsection{E7: a static benchmark, MMLU}

\textbf{Data.} Stanford HELM publishes per-instance outputs of its MMLU leaderboard (release v1.13.0) \citep{liang2023helm,helm_data,hendrycks2021mmlu}. We use the per-item exact-match results of four open-weight models on the $14{,}042$ questions of the released MMLU test set, a multiple-choice exam in 57 subjects. The models are Llama-3.1-8B-Instruct, Qwen2.5-7B-Instruct, Gemma-2-9B and Mistral-7B-Instruct-v0.3. The MMLU paper reports $14{,}079$ test questions. Each prompt contains five solved examples (5-shot). HELM ran Llama-3.1-8B-Instruct and Qwen2.5-7B-Instruct through the Turbo endpoints of Together AI, a cloud inference provider. These endpoints use 8-bit floating-point (FP8) quantization. Gemma-2-9B is the base model, since HELM's MMLU leaderboard does not include the instruction-tuned one. The four models' accuracies over the $14{,}042$ questions are $0.542$, $0.714$, $0.704$ and $0.589$. HELM's leaderboard averages over the $57$ subjects instead and reports $0.561$, $0.729$, $0.721$ and $0.599$. A record is one question. For the ordered pair $(i,j)$, $W=1$ if $i$ is correct and $j$ wrong, $W=0$ for the reverse, and $W=\frac12$ if both are correct or both wrong. A contaminated item is one whose record contamination altered, whatever the cause. The record order is a random permutation drawn before any outcome is looked at (seed $0$). The permutation only defines an online order and does not by itself make the records independent. For the population reading of a claim we adopt a modeling assumption, that the questions are an i.i.d.\ sample from a larger population of questions. A claim then compares the two models on that population. Without this assumption, \cref{prop:finite} still gives the tolerance of the released benchmark itself, reported below. By \eqref{eq:netwins} of \cref{app:testing}, the final tolerance does not depend on the order, and we confirmed that it was identical over $100$ random orders for every pair. The time to certify does depend on the order, and \cref{tab:mmlu} reports it for the seed-$0$ order and its median over the $100$ orders.

\textbf{Tolerances.} \Cref{tab:mmlu} gives the tolerances. The four well-separated pairs ($\Delta$ above $1{,}600$) are certified against $1{,}469$ to $2{,}264$ forged items and $714$ to $1{,}102$ contaminated items. The forged-item tolerances are $91$ to $94\%$ of $\Delta-1$, the number of forged items that the released benchmark itself withstands (\cref{prop:finite}), and the contaminated-item tolerances are $89$ to $91\%$ of $\lceil\Delta/2\rceil-1$, its count for replaced items. Contaminating one model moves a record by at most $\frac12$, so for contamination alone the deterministic count is $\Delta-1$ and these tolerances are conservative. The lead certificate gives $1{,}269$ to $2{,}081$ forged and $626$ to $1{,}026$ contaminated items on the four well-separated pairs, $7$ to $14\%$ below the edge certificate. Mistral-7B over Llama-3.1-8B is certified against $470$ forged items, $71\%$ of $\Delta$. Qwen2.5-7B against Gemma-2-9B is not certified in either direction. Their net difference is $140$ wins on $2{,}960$ untied items.

\textbf{Rigging replay.} The replay follows E3b. For each of the six pairs the promoted model is the one with the lower accuracy, so the claim ``the promoted model is better'' is false. Null streams resample the pair's own items with replacement. The attacker inserts $B$ forged items ($W=1$) in one burst at a random position, spreads them at rate $B/n$, or flips the first $B$ items that the promoted model lost ($W=0\to1$, the value-dependent attacker). The methods are those of \cref{sec:experiments}, with $1{,}000$ replays per cell. \Cref{fig:benchmark}(a,b) shows the mean over the six pairs. Forged bursts of up to $400$ items leave our certificates at or below $0.006$ on every pair. The Wald interval and undeflated betting reach $0.36$ and $0.33$ on average over the pairs, and $0.999$ and $0.995$ on the worst pair. Flipping the lost items breaks the Wald interval from $5$ items ($0.35$, and $0.92$ at $10$) and the per-step robust test at $25$ ($0.86$). The insertion certificate fails from $10$ items ($0.08$) and reaches $0.44$ at $25$ and $0.92$ at $50$, as \cref{thm:sharp}(b) predicts. The replacement certificate stays at or below $0.008$ on every pair and budget.

\begin{table*}[htbp]
\caption{E7: tolerances on HELM MMLU ($14{,}042$ questions) for the six pairs, oriented by accuracy. Each row concerns the edge and lead claims for the first model against the second. Edge and lead: the edge tolerance of \eqref{eq:bhat-grid} and the lead tolerance of \cref{thm:average} after all questions, in forged and in contaminated items. Wald: the items that the fixed-time $95\%$ Wald interval absorbs (E10). $\Delta$: wins minus losses. $\tau_0$: questions until the claim is first certified, in the seed-$0$ order (seed 0) and as a median over $100$ random orders (median). --: not certified at the end, or never certified ($\tau_0$, in the seed-$0$ order or in most of the $100$ orders).}
\label{tab:mmlu}
\centering\scriptsize\setlength{\tabcolsep}{3pt}
\begin{tabular}{lrrrrrrrrrrrr}
\toprule
 & & & & & \multicolumn{2}{c}{edge} & \multicolumn{2}{c}{lead} & \multicolumn{2}{c}{Wald} & \multicolumn{2}{c}{$\tau_0$} \\
\cmidrule(lr){6-7}\cmidrule(lr){8-9}\cmidrule(lr){10-11}\cmidrule(lr){12-13}
pair & wins & losses & ties & $\Delta$ & forged & contam. & forged & contam. & forged & contam. & seed 0 & median \\
\midrule
Qwen2.5-7B $>$ Gemma-2-9B & 1,550 & 1,410 & 11,082 & 140 & -- & -- & -- & -- & 33 & 16 & -- & -- \\
Qwen2.5-7B $>$ Mistral-7B & 2,841 & 1,087 & 10,114 & 1,754 & 1,607 & 781 & 1,411 & 696 & 1,660 & 815 & 91 & 141 \\
Qwen2.5-7B $>$ Llama-3.1-8B & 3,571 & 1,156 & 9,315 & 2,415 & 2,264 & 1,102 & 2,081 & 1,026 & 2,318 & 1,140 & 61 & 95.5 \\
Gemma-2-9B $>$ Mistral-7B & 2,672 & 1,058 & 10,312 & 1,614 & 1,469 & 714 & 1,269 & 626 & 1,521 & 747 & 204 & 160.5 \\
Gemma-2-9B $>$ Llama-3.1-8B & 3,532 & 1,257 & 9,253 & 2,275 & 2,117 & 1,032 & 1,939 & 956 & 2,174 & 1,069 & 134 & 101.5 \\
Mistral-7B $>$ Llama-3.1-8B & 2,531 & 1,870 & 9,641 & 661 & 470 & 229 & 303 & 149 & 539 & 265 & 1,281 & 1,195 \\
\bottomrule
\end{tabular}
\end{table*}

\begin{figure*}[htbp]
\centering
\includegraphics[width=\textwidth]{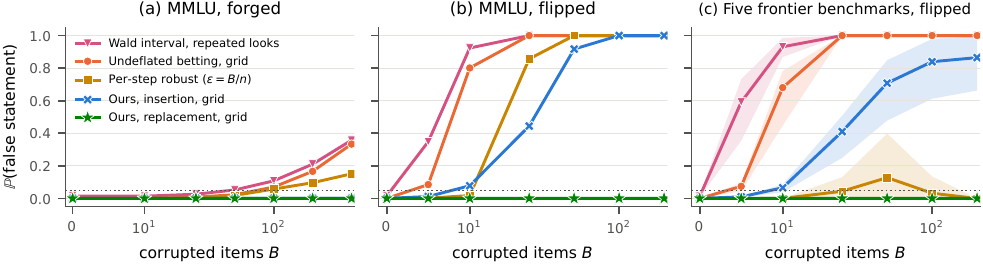}
\caption{E7 and E8: probability of a false statement on resampled benchmark items ($1{,}000$ replays per cell, $\alpha=0.05$ dotted). (a,b) MMLU, mean over the six pairs of four open-weight models with 7 to 9 billion parameters. Forged items are inserted in one burst at a random position (a). Flipped items are the first $B$ items the promoted model lost, turned into wins (b). (c) Flipped items on the five frontier benchmarks of E8. Each line averages the five benchmarks, each benchmark a mean over its pairs. Each band is the range over the five benchmarks.}
\label{fig:benchmark}
\end{figure*}

\subsection{E8: frontier models on HELM and SWE-bench Verified}

\textbf{HELM data.} HELM's capabilities leaderboard (release v1.15.0) evaluates $68$ models \citep{helm_data}. We use fifteen frontier models of 2025 from seven developers, chosen by developer and release from HELM's model list. They are Claude Opus 4, Claude Sonnet 4, Claude Sonnet 4.5, o3, GPT-4.1, GPT-5, GPT-5.1, gpt-oss-120b, Gemini 2.5 Pro (03-25 preview), Gemini 3 Pro (preview), Grok 4, DeepSeek-R1-0528, Kimi K2, Qwen3-235B and Qwen3-235B-2507. The exact HELM identifiers are in the repository. HELM ran the Claude models without extended thinking, their optional reasoning mode, and both Qwen3 models with FP8 quantization. HELM ran GPT-5.1 with the provider's default setting, which turns reasoning off. Its accuracy on MMLU-Pro and GPQA ($0.579$ and $0.442$) is therefore far below that of GPT-5 ($0.863$ and $0.791$). The benchmarks are a sample of $1{,}000$ questions of MMLU-Pro, a harder version of MMLU \citep{wang2024mmlupro}, and $446$ of the $448$ graduate-level science questions of GPQA main (Graduate-Level Google-Proof Q\&A) \citep{rein2023gpqa}. They also include the $541$ prompts of IFEval (Instruction-Following Evaluation), with instructions whose compliance is checked automatically \citep{zhou2023ifeval}, and a sample of $1{,}000$ olympiad mathematics problems of Omni-MATH \citep{gao2025omnimath}. On MMLU-Pro and GPQA the prompt has no solved examples and asks for step-by-step reasoning (zero-shot chain of thought). HELM publishes per-instance outputs. On MMLU-Pro and GPQA it records whether each answer is correct. On IFEval it records the fraction of a prompt's instructions that were followed, and we count a prompt as right only when all were followed (prompt-level strict accuracy). On Omni-MATH it records the share of language-model judges that accept the answer ($0$, $\frac13$, $\frac12$, $\frac23$ or $1$), and $4.2\%$ of these scores are fractional. The record is then $W=\frac12(1+s_i-s_j)$ for the scores $s_i$ and $s_j$ of the two models. This equals the record of E7 when both scores are $0$ or $1$, and \cref{thm:validity} allows any $W\in[0,1]$.

\begin{figure*}[htbp]
\centering
\includegraphics[width=\textwidth]{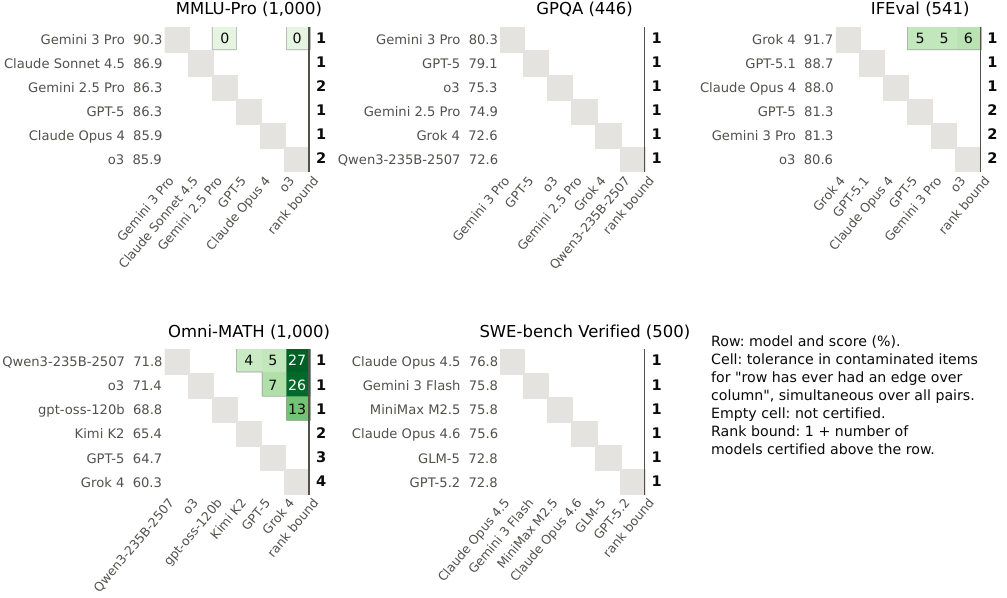}
\caption{E8: the six highest-scoring models of each frontier benchmark, with their scores in percent. A cell gives the simultaneous tolerance of the edge claim ``row model has ever had an edge over column model'' in contaminated items, by closed testing over all ordered pairs of the benchmark's models. An empty cell means not certified, and $0$ means certified against no contaminated item. The last column is the certified rank bound, one plus the number of models certified above the row model.}
\label{fig:frontier}
\end{figure*}

\textbf{SWE-bench data.} A SWE-bench task asks a model to resolve a real issue of a GitHub repository, and it counts as resolved when the repository's tests pass \citep{jimenez2024swebench}. SWE-bench Verified is a subset of $500$ tasks screened by software developers \citep{openai2024swebenchverified}. Its maintainers publish runs of one simple agent with many models and the outcome of every task \citep{swebench_experiments}. The agent is mini-SWE-agent, a minimal agent of the SWE-agent project \citep{yang2024sweagent}, in which the model acts only through shell commands. We use the runs of its version 2 with eleven models of 2025 and 2026, all run with version 2.0.0. They are Claude Opus 4.5 and 4.6, Claude Sonnet 4.5, Claude Haiku 4.5, GPT-5.2, GPT-5 mini, Gemini 3 Flash, GLM-5, Kimi K2.5, DeepSeek V3.2 and MiniMax M2.5. A run is kept only if its per-task file has all $500$ tasks and agrees with its reported score. Three runs fail this check. The file of Gemini 3 Pro marks every task unresolved although $69.6\%$ is reported. GPT-5.2-Codex has no per-task file, and the file of Gemini 3.5 Flash covers $441$ tasks. On all five benchmarks the records are those of E7, with the graded record above on Omni-MATH, and the order and the population assumption are those of E7. As in E7, the final tolerance was identical over $100$ random orders for every ordered pair with $\Delta>0$. The model sets are fixed by name and release, not by the records, so the family of all ordered pairs of a benchmark's models does not depend on the data. Each item is one time step at which every pair receives its record, and a contaminated item counts once however many pairs it changes. \Cref{cor:items} in \cref{app:proofs} extends \cref{prop:closed} to this setting, with the replacement cost also against a developer who knows every item in advance, and \cref{thm:validity} applies to each pair directly. The five benchmarks were released in 2023 and 2024. The HELM results come from release v1.15.0 of November 2025 and the SWE-bench Verified runs from February 2026.

\textbf{Tolerances.} \Cref{tab:frontier} summarizes the tolerances. On the HELM benchmarks, $55$ to $88$ of the $210$ ordered pairs are certified pair by pair, and $36$ to $73$ simultaneously by closed testing. On SWE-bench Verified the counts are $25$ and $14$ of $110$, and all $14$ involve one of the two weakest models. On every benchmark neither of the two highest-scoring models is certified above the other simultaneously. Within the top five, the largest simultaneous tolerance is $17$ forged or $7$ contaminated items (o3 over GPT-5 on Omni-MATH). On MMLU-Pro and GPQA the largest simultaneous tolerances are against GPT-5.1 without reasoning, up to $279$ forged items. No model among the six best on SWE-bench Verified ($72.8$ to $76.8\%$) is certified above another even pair by pair. \Cref{fig:frontier} shows the simultaneous tolerances among the six highest-scoring models of each benchmark. Arena reports an approximate rank for each model, one plus the number of models whose confidence interval lies above its own \citep{chiang2024arena}. \Cref{fig:frontier} also gives the certified analogue, one plus the number of models certified above a model. These rank bounds hold for all models at once with probability at least $1-\alpha$ when no item is corrupted (\cref{cor:items}). On GPQA and SWE-bench Verified the bound is $1$ for each of the six models.

\textbf{Rigging replay.} The replay of E7 runs on every pair with unequal scores, $53$ to $105$ per benchmark. Five flipped items break the Wald interval with probability $0.35$ to $0.73$ on average over the pairs, and ten with probability $0.87$ to $0.99$. On average over the pairs, the insertion certificate fails from $10$ to $25$ flipped items, as \cref{thm:sharp}(b) predicts. Under flips the replacement certificate stays at or below $0.015$ on every pair and budget (means in \cref{fig:benchmark}(c)). Twenty-five forged items in one burst push the Wald interval to $0.15$ to $0.48$ on average. Under bursts of up to $400$ items, our certificates stay at or below $0.021$ on every pair.

\begin{table*}[htbp]
\caption{E8: summary per frontier benchmark. Items and models: the benchmark's size. Certified pairs: ordered pairs certified at $\alpha=0.05$, pair by pair or simultaneously by closed testing over all ordered pairs of the benchmark's models. Top two: simultaneous tolerance of the leader over the runner-up. Largest tolerance: the largest simultaneous tolerance among the pairs within the top five, and among all pairs. Tolerances are forged~/~contaminated items, and -- means not certified. The top two and top five are the first rows of \cref{fig:frontier}. On MMLU-Pro, Claude Opus 4 and o3 tie for fifth at $85.9\%$, and with o3 as fifth the entry is 2~/~0.}
\label{tab:frontier}
\centering\scriptsize
\resizebox{\textwidth}{!}{\begin{tabular}{lrrrrccc}
\toprule
benchmark & items & models & certified pairs & certified pairs & top two & largest tolerance & largest tolerance \\
 & & & (pair by pair) & (simultaneously) & & in the top five & overall \\
\midrule
MMLU-Pro & 1,000 & 15 & 60 of 210 & 42 & -- & 1 / 0 & 279 / 117 \\
GPQA & 446 & 15 & 55 of 210 & 36 & -- & -- & 118 / 47 \\
IFEval & 541 & 15 & 67 of 210 & 38 & -- & 15 / 5 & 88 / 33 \\
Omni-MATH & 1,000 & 15 & 88 of 210 & 73 & -- & 17 / 7 & 272 / 111 \\
SWE-bench Verified & 500 & 11 & 25 of 110 & 14 & -- & -- & 69 / 26 \\
\bottomrule
\end{tabular}}
\end{table*}

\subsection{E9: the lead tolerance on i.i.d.\ records}

Uncorrupted records are i.i.d.\ with mean outcome $\frac12+\delta$, and there is no attack. A record is a win with probability $\pi/2+\delta$, a loss with probability $\pi/2-\delta$ and a tie otherwise, with $\pi=1$ (binary records) or $\pi=0.7$ ($30\%$ ties, close to the $36\%$ on Chatbot Arena). For $\delta\in\{0.05,0.1\}$, $2{,}000$ repetitions of $5{,}000$ records are drawn once. The edge tolerance of \eqref{eq:bhat-grid} and the lead tolerance of \cref{thm:average} are computed on the same records with both costs and the grid of \cref{sec:method}. \Cref{tab:average} reports $\Bhat_t/t$ at $t=5{,}000$ and the limit rates of \cref{thm:upper} and \cref{thm:average}(d). For a binary outcome, $\log(1+\gamma(W_s-\frac12))=\gamma(W_s-\frac12)-\gamma^2/8+O(\gamma^3)$, so on binary records the two wealths differ by $O(\gamma^3t)$. Accordingly the two tolerances differ in certified status on $0.01\%$ of the (repetition, record) cells, and their medians at $t=5{,}000$ agree to four decimals. A tie leaves $K^{(k)}_t$ unchanged and multiplies $Z^{(\gamma)}_t$ by $e^{-\gamma^2/8}$. With $30\%$ ties the lead tolerance is therefore $4$ to $11\%$ lower at $t=5{,}000$, while the limit rates differ by less than $1\%$. At $\delta=0.1$ the limit rates are $0.198$ to $0.199$ against the optimum $0.2$ for insertion and $0.099$ against $0.1$ for replacement.

\begin{table*}[htbp]
\caption{E9: tolerance per record $\Bhat_t/t$ at $t=5{,}000$ on i.i.d.\ records with mean outcome $\frac12+\delta$ and no attack ($2{,}000$ repetitions; median and $10$--$90\%$ range). Ties: share of tied records. Edge: the tolerance of \eqref{eq:bhat-grid}. Lead: the lead tolerance of \cref{thm:average}. Limit rates: the almost-sure limits of $\Bhat_t/t$ given by \cref{thm:upper} and \cref{thm:average}(d). The optimal rates are $2\delta$ for insertion and $\delta$ for replacement.}
\label{tab:average}
\centering\small
\begin{tabular}{llcccc}
\toprule
$\delta$ & ties & cost & edge \eqref{eq:bhat-grid} & lead \eqref{eq:zwealth} & limit rates (edge, lead) \\
\midrule
0.05 & 30\% & ins. & 0.066 (0.051--0.081) & 0.059 (0.043--0.073) & 0.098, 0.098 \\
0.05 & 30\% & rep. & 0.032 (0.025--0.039) & 0.029 (0.021--0.036) & 0.049, 0.049 \\
0.05 & 0\% & ins. & 0.058 (0.039--0.077) & 0.058 (0.039--0.077) & 0.098, 0.098 \\
0.05 & 0\% & rep. & 0.028 (0.019--0.037) & 0.028 (0.019--0.037) & 0.049, 0.049 \\
0.10 & 30\% & ins. & 0.169 (0.153--0.184) & 0.160 (0.145--0.176) & 0.199, 0.198 \\
0.10 & 30\% & rep. & 0.082 (0.074--0.089) & 0.078 (0.071--0.086) & 0.099, 0.099 \\
0.10 & 0\% & ins. & 0.160 (0.141--0.178) & 0.160 (0.141--0.178) & 0.198, 0.198 \\
0.10 & 0\% & rep. & 0.078 (0.069--0.087) & 0.078 (0.069--0.087) & 0.099, 0.099 \\
\bottomrule
\end{tabular}
\end{table*}

\subsection{E10: fixed-time breakdown counts}

For each pair of E3a and E7, the columns ``Wald'' of \cref{tab:arena-pairs,tab:mmlu} give a breakdown count. It is the largest number of forged or flipped records that the two-sided $95\%$ Wald interval for the mean outcome, computed once on all records, absorbs while still excluding one half. This is what a sensitivity study of the final standings computes \citep{huang2025voting,oyarhoseini2026perturbation}. It is a fixed-time quantity, and \cref{tab:rigging} and \cref{fig:benchmark}(a,b) show the same interval certifying false orderings under adaptive timing. On the Arena pairs certified above $500$ votes, the edge tolerance in forged votes is $90$ to $97\%$ of the Wald count (median $95\%$) and the lead tolerance $82$ to $94\%$. On the four well-separated MMLU pairs the edge tolerance is $97$ to $98\%$ and the lead tolerance $83$ to $90\%$. The Hoeffding interval, with half-width $\sqrt{\log(2/\alpha)/(2n)}$, breaks down below the edge tolerance on every certified pair. Two pairs with a positive Wald count are not certified: Gemini-1.5-Pro-0409 over Claude-3-Opus ($19$ forged votes) and Qwen2.5-7B over Gemma-2-9B ($33$ items). The anytime guarantee thus costs a few percent of the breakdown count of a fixed-time interval that attacks can time.

\FloatBarrier

\section{Additional Discussion}\label{app:discussion}

\subsection{The guarantee in the language of hypothesis testing}\label{app:testing}

\Cref{def:budget} controls the type I error of a family of sequential tests, one for each possible number of corrupted records. For each integer $b\ge0$ the null hypothesis at time $t$ is that $H_0^{ij}$ holds and at most $b$ of the first $t$ records were corrupted, and the platform rejects it when $\Bhat_t\ge b$. Thus $\Bhat_t$ is the largest number of corrupted records for which the null hypothesis is rejected at time $t$. The second part of each null hypothesis, $N_t\le b$, is a random event and not a property of the joint distribution, so these are hypotheses in an extended sense. The formal guarantee is \eqref{eq:validity}.

A type I error is the rejection of a null hypothesis that is true. If $H_0^{ij}$ fails, no null hypothesis in the family is true, and no type I error can occur. If $H_0^{ij}$ holds, the null hypotheses that are true at time $t$ are those with $b\ge N_t$, and one of them is rejected exactly when $\Bhat_t\ge N_t$. Hence \eqref{eq:validity} states that, when $H_0^{ij}$ holds, the probability of a type I error at any time $t$ and for any $b$ is at most $\alpha$. Equivalently, the probability that any published statement is ever false is at most $\alpha$. The platform does not observe $N_t$. It therefore publishes the conditional statement ``$i$ has ever had an edge over $j$, unless more than $\Bhat_t$ of the first $t$ records were corrupted'' and not a rejection of $H_0^{ij}$ itself.

For the tolerance of \eqref{eq:bhat-grid}, the quantity $1/\Phi_t(b)$ plays the role of a p-value for the null hypothesis that $H_0^{ij}$ holds and at most $b$ records were corrupted: $\Bhat_t\ge b$ exactly when $1/\Phi_t(b)\le\alpha$. More precisely, $\Phi_t(b)$ does not depend on $\alpha$ and \cref{thm:validity} holds at every level. So for every $\alpha'\in(0,1)$ and every attacker in a class that \cref{thm:validity} covers with the chosen cost, when $H_0^{ij}$ holds, $\P\big(\exists\,t,\ \exists\,b\ge N_t:\ 1/\Phi_t(b)\le \alpha'\big)\le \alpha'$. Indeed, if $b\ge N_t$ and $\Phi_t(b)\ge1/\alpha'$, then $\Phi_t(N_t)\ge\Phi_t(b)\ge1/\alpha'$ because $\Phi_t$ is decreasing. In that case the tolerance computed at level $\alpha'$ is at least $N_t$, an event of probability at most $\alpha'$ over all $t$.

\paragraph{What a correct edge claim says.} Suppose that there are no ties and that $i$ wins each of the first $1{,}000$ records with probability $0.6$ and each later record with probability $0.4$. Then $H_0^{ij}$ is false and the edge claim is correct in the sense of \cref{sec:formulation}. The usual reading of a leaderboard claim is the stable case of \cref{sec:formulation}, in which $\E[X_t\mid\cG_{t-1}]$ is the same constant $\mu$ at every non-inserted record and the claim is correct exactly when $\mu>\frac12$. \Cref{thm:upper,thm:lower} assume i.i.d.\ records, a special case of the stable case. The simulations of \cref{sec:experiments} also use i.i.d.\ uncorrupted values. \Cref{thm:validity} needs neither. In the example the uncorrupted margin of \eqref{eq:delta} is $\Delta_t=0.1\min(t,1{,}000)-0.1(t-1{,}000)^+$, positive before $t=2{,}000$ and at most zero from then on. So the lead claim of \cref{sec:formulation} is correct exactly before $t=2{,}000$, and \cref{thm:average} bounds by $\alpha$ the probability that the lead tolerance is nonnegative at any later time (E5 in \cref{app:experiments}). Unlike $H_0^{ij}$, the event $\{\Delta_t\le0\}$ is random and may hold at some times and not at others. The lead certificate reads ``$i$ leads $j$, unless more than $\Bhat_t$ of the first $t$ records were corrupted''.

The tolerance itself does respond to later records. Let $n^{\rm win}_t$ and $n^{\rm loss}_t$ be the numbers of wins and losses of $i$ among the first $t$ records. For outcomes in $\{0,\frac12,1\}$,
\begin{equation}\label{eq:netwins}
\log K^{(k)}_t=\tfrac12\,\crep(\lb_k)\,\big(n^{\rm win}_t-n^{\rm loss}_t\big)-\tfrac12\,\kappa(\lb_k)\,\big(n^{\rm win}_t+n^{\rm loss}_t\big),
\end{equation}
with $\kappa$ as in \eqref{eq:growth}, because $\log(1+\lb/2)=(\crep(\lb)-\kappa(\lb))/2$ and $\log(1-\lb/2)=-(\crep(\lb)+\kappa(\lb))/2$ and a tie contributes $\log1=0$. Hence $\Phi_t(b)$ and $\Bhat_t$ depend on the records only through the numbers of wins and losses so far. In the example above every $K^{(k)}_t$ tends to zero and $\Bhat_t$ returns to $-1$. In simulation it does so at a median of $420$ records after the switch (E5 in \cref{app:experiments}). The platform may also start the count at a time fixed in advance. \Cref{def:budget} and \cref{thm:validity}(a,b) then apply to the records from that time on and the claim concerns those records only. E3c in \cref{app:experiments} does this for the last $120$ days of the Arena data. The lead claim is certified by \cref{thm:average} with the exponential wealth of \cref{sec:expwealth}.
\subsection{Why this is the right formulation}\label{app:why}

Attacks and defenses are both counted in records, from vote-rigging and sensitivity studies of Arena \citep{huang2025voting,min2025rigging,huang2025dropping,oyarhoseini2026perturbation} to certified defenses in machine learning, sequential decision making, distributed systems and robust statistics (\cref{app:fields}).

Anytime validity and an adaptive attacker are required by how leaderboards are used. A statement valid only at one planned sample size loses its guarantee under repeated checking \citep{howard2021,khosravi2026rank}. An attacker who sees the standings can wait and strike in one burst, which a per-step bound on corruption cannot describe (\cref{prop:perstep}).

\subsection{How adjacent fields count corruption}\label{app:fields}

Adjacent fields measure attacks and defenses by the amount of corruption, usually a count of corrupted records or participants. Vote-rigging and sensitivity studies of Arena report the number of votes that move a model or change the top model \citep{huang2025voting,min2025rigging,huang2025dropping,oyarhoseini2026perturbation}. Certified defenses against data poisoning (altering the training data to change a trained model) approximately bound the damage that a given number of poisoned points can do \citep{steinhardt2017certified}. Others certify predictions against a number of altered points \citep{levine2021dpa,jia2021bagging}. \citet{bose2025dynamic} certify an online learner when each record is poisoned with a fixed probability and an adaptive attacker picks its value, a per-step model as in \cref{prop:perstep}. Corruption-robust bandits, sequential decision rules whose rewards an adversary may alter, pay a price linear in the total corruption \citep{lykouris2018,gupta2019}. Byzantine-robust systems, which tolerate participants who behave arbitrarily, count the corrupted participants \citep{lamport1982,blanchard2017}. In robust statistics, \emph{strong contamination} lets an adversary inspect the sample and replace a fraction of it, whereas \emph{Huber contamination} mixes in outliers before sampling \citep{huber1964,diakonikolas2023book}. Our classes $\classV$ and $\classP$ are their online counterparts, with the count $N_t$ in place of the fraction. An attacker in $\classO$ sees all uncorrupted values before acting, as the adversary of strong contamination does.

\textbf{Leaderboard inference assumes genuine votes.} \citet{chiang2024arena} flag anomalous Arena users with a heuristic sequential test and find that versions built on nonnegative supermartingales and e-values lacked power. Arena's confidence intervals for Bradley--Terry ratings \citep{bradley1952,chiang2024arena} are valid only at a fixed sample size. So are rank confidence sets \citep[sets of ranks containing a model's true rank with a stated probability;][]{mogstad2024} and inference on selected winners \citep{andrews2024winners}. Robust rank aggregation from corrupted pairwise comparisons \citep{agarwal2020rank} bounds the estimation error, not the error of published claims. Anytime-valid rank sets \citep{khosravi2026rank} assume uncorrupted data. They combine pairwise e-processes by \emph{closed testing} \citep{marcus1976,fischer2024closed}, which tests intersections of hypotheses to control the probability of any false rejection.

\textbf{Attack studies and certified defenses give no sequential certificate.} The attack and sensitivity studies of \cref{sec:intro,app:why} \citep[see also][]{zhao2024challenges} show that attacks are real and cheap but give no guarantee. \citet{ma2024sequential} fabricate pairwise comparisons online, without knowledge of future data, to manipulate rank aggregation; our certificates bound how many such fabrications a published claim withstands. Certified defenses against data poisoning certify one classifier or bound a long-run average loss, not a sequence of published statements \citep{steinhardt2017certified,levine2021dpa,bose2025dynamic}. Corruption-robust bandits bound regret, the reward lost to the best fixed choice \citep{lykouris2018,gupta2019}, or the error of one final choice \citep{zhong2021psr,altschuler2019contaminated}, not errors at every time.

\subsection{Limitations and extensions}\label{app:limitations}

The tolerance counts records, not attackers. One attacker can control many cheap accounts, so the tolerance should be read alongside the cost of an account. The null hypothesis treats each pair separately and ignores the transitivity (if $i$ is better than $j$ and $j$ than $j'$, then $i$ is better than $j'$) that Bradley--Terry models exploit. Our replays simulate attacks, since real attack logs are not public. Validity against offline attackers when the bets change with the data remains open. So do withholding attacks, which suppress uncorrupted records. For attacks on correct claims, \cref{prop:power} bounds the tolerance lost per corrupted record and the time to recover it. Empirically, each inserted loss removes about one forged-vote unit of tolerance. Each win replaced by a loss removes about two, which is one flipped-vote unit. So de-certifying a claim takes about as many inserted losses as the forged-vote tolerance, or as many replaced wins as the flipped-vote one (E6 in \cref{app:experiments}).

\end{document}